\documentclass[a4paper,fleqn]{cas-sc}

\usepackage[authoryear,longnamesfirst]{natbib}

\newtheorem{theorem}{Theorem}
\newtheorem{lemma}[theorem]{Lemma}

\newtheorem{corollary}[theorem]{Corollary}
\newdefinition{remark}{Remark}
\newdefinition{definition}{Definition}
\newdefinition{example}{Example}
\newdefinition{problem}{Problem}
\newproof{proof}{Proof}

\usepackage{graphicx}%
\usepackage{multirow}%
\usepackage{amsmath,amssymb,amsfonts}%
\usepackage{mathrsfs}%
\usepackage{dsfont}%
\usepackage{subcaption} 
\usepackage{enumitem}  

\newcommand{\R}{\mathbb R}
\newcommand{\N}{\mathbb N}

\newcommand{\ep}{\varepsilon}

\newcommand{\Fourier}[1]{\mathcal{F}\left(#1\right)}
\newcommand{\one}{\mathds{1}}
\newcommand{\vp}{\mathrm{pv}\frac{1}{x}}

\newcommand{\pp}{\tilde{p}}

\newcommand{\Poly}{\mathscr{P}}
\newcommand{\Sw}{\mathscr{S}}
\newcommand{\SM}{\mathscr{S}_M}

\newcommand{\dotProduct}[2]{\langle #1, #2 \rangle}

\DeclareMathOperator{\ind}{ind}
\DeclareMathOperator{\sinc}{sinc}
\DeclareMathOperator{\sign}{sign}

\DeclareMathOperator{\supp}{supp}

\DeclareMathOperator{\id}{id}
\DeclareMathOperator{\dom}{dom}

\begin{document}
\let\WriteBookmarks\relax
\def\floatpagepagefraction{1}
\def\textpagefraction{.001}

\shorttitle{Exact Deconvolution via Recursive Convolutions}

\shortauthors{A. González-Calvin}  

\title [mode = title]{Exact Deconvolution for Schwartz Kernels: From Polynomial Automorphisms to Recursive Inversion in Tempered Distributions}  


%

\author[1]{Alfredo González-Calvin}[
orcid=0009-0004-2760-7096]

\cormark[1]


\ead{alfredgo@ucm.es}



\affiliation[1]{organization={Dept. of Computer Architecture and Automation, 
Faculty
        of Physics, Complutense University of Madrid},
    addressline={Plaza de las Ciencias, 1}, 
    city={Madrid},
    postcode={28040}, 
    state={Madrid},
    country={Spain}}

\cortext[1]{Corresponding author}



\begin{abstract}
In this work, we construct an explicit,
theoretically rigorous deconvolution method that relies entirely on
iterative forward convolutions, thus can be numerically implemented.
We first prove that convolution with an even
Schwartz kernel acts as an automorphism on the vector space of finite-degree
polynomials. Exploiting the parity of the kernel, we derive an exact
algebraic inverse for this space, expressed uniquely as a finite linear
combination of repeated convolutions. 
The core contribution of this work extends this algebraic inversion to
infinite-dimensional function spaces, including $L^1(\R)$, $L^2(\R)$, the
Schwartz space $\Sw(\R)$, and the space of tempered distributions
$\Sw'(\R)$. By passing the finite-sum polynomial inversion formula
to the limit, we demonstrate that an arbitrary function or distribution
convolved with a Schwartz kernel can be exactly recovered in its respective
topology. The resulting inverse is an explicitly computable limit of a
sequence of linear combinations of recursive convolutions. As a primary
application, this limit provides a fundamentally new, iterative numerical
formula for the inverse of the Weierstrass Transform. By bypassing
traditional numerically ill-posed inversion techniques, our method offers a
mathematically exact and numerically robust algorithm for computational signal
recovery. 
\end{abstract}


\begin{highlights}
    \item Convolution with an even Schwartz kernel is a polynomial automorphism.
    
    \item Exact polynomial deconvolution via finite iterative forward convolutions.
    
    \item Recursive limits extend exact inversion to $L^1$, $L^2$, and distributions.
    
    \item Bypasses ill-posed techniques, recovering functions in their topologies.
    
    \item Yields a robust iterative numerical inverse for the Weierstrass transform.
\end{highlights}


\maketitle

\section{Introduction}

The Weierstrass transform with parameter $\ep > 0$ of a function $f : \R \to \R$, is
given, when it exists, by
\begin{equation*}
    F_{\ep}(x) := (f * \varphi_{\ep})(x) = \int_{\R}f(x-s)\varphi_{\ep}(s)ds,
\end{equation*}
where $x\mapsto \varphi_{\ep}(x) :=
\frac{1}{\ep\sqrt{4\pi}}e^{-\frac{x^2}{4\ep^2}}$ is the so called
\textit{Gaussian Kernel}. The Gaussian Kernel is an even Schwartz function that
satisfies several properties. It has unit area, i.e.,
$\int_{\R}\varphi_{\ep} = 1$, it converges to the Dirac delta distribution as
$\ep \to 0^+$, and for any $f \in L^p_{loc}(\R)$, $F_{\ep} \in C^{\infty}(\R)$
and $F_{\ep} \to f$ as $\ep \to 0^+$ in several types of convergence. In
particular if $f \in L^p_{loc}(\R)$ then $F_{\ep} \to f$ as $\ep \to 0^+$
pointwise and in $L^p$. The theoretical foundation of this transform has been
studied extensively in \cite{widder1951necessary}, where its inversion is
presented as a differential operator, and in \cite{bilodeau1962weierstrass},
which explores the relationship between the transform and Hermite polynomials.

While classical inversion formulas provide elegant theoretical solutions, they
often rely on unbounded differential operators or Fourier domain
divisions that can be highly susceptible to numerical instability and noise.
Observing the action of the Weierstrass transform on simple spaces like
polynomials, it is natural to ask how polynomials behave under a more general
convolution kernel, and crucially, if there is an inverse for this operation
which is both mathematically rigorous and simple to compute numerically which
can be extended to a bigger space of functions. That
is, instead of focusing solely on the Gaussian Kernel, we generalize our study
to convolution with arbitrary Schwartz kernels. These questions possess deep
computational applications, since convolution with such a function models
weighted averages, low-pass filters, and various physical phenomena.
Consequently, the process of \textit{deconvolution} --- the restoration of an
original signal from a filtered one --- is an important problem in
computational and applied mathematics.

The theoretical and computational importance of deconvolution is
well-established across applied mathematics. Classical treatments such as
\cite{berenstein1983some} and \cite{hirschman2012convolution} provide systematic
inversion formulas for various convolution transforms, while explicit solutions
to convolution integral equations have historically relied on differentiation
\cite{hohlfeld1993solution} or Fourier domain restrictions
\cite{prost2003deconvolution}. To overcome the numerical vulnerabilities of
these classical analytical methods, modern approaches often utilize limiting
sequences \cite{baeumer2003inversion}, geometric series methods for arbitrary
distribution functions
\cite{kaiser2025deconvolutionarbitrarydistributionfunctions}, and regularized
iterative filtering techniques \cite{LI2018425}. Furthermore, in
finite-dimensional discrete settings, structural algebraic frameworks --- such
as polynomial-based matrix approaches --- have been successful in performing
exact deconvolution \cite{DIAZTOCA201798}. However, when transitioning to
practical numerical implementations, carefully managing convolution matrices and
boundary conditions remains a significant computational challenge
\cite{ZHOU201414, Hansen2002-li}. This underlying complexity has prompted
comparative numerical analyses \cite{madden1996comparison} and continues to
drive modern blind deconvolution techniques \cite{campisi2017blind,
zhang2023fast, XU2025116511, DYKES2021113067} and broader algorithmic frameworks
\cite{Hohage_Marechal_Simar_Vanhems_2024} across diverse physical and medical
applications \cite{ulmer2010inverse, fanton2021convolution, HOWARD201443}.

In this work, we contribute to this field by developing a
novel, computationally viable framework for exact deconvolution. In Section
\ref{sec:TheoreticalResults}, we prove that convolution with an even Schwartz
function is an invertible operator (an automorphism) on the space of
finite-order polynomials (Theorem
\ref{thm:ConvolutionIsAnIsomorphismInPolynomials}). Then, instead 
of taking an algebraic perspective, we also show that
the exact inverse of this operator can be expressed purely as a finite linear
combination of iterative convolutions (Equations \eqref{eq:InverseOfConvolution}
and \eqref{eq:InverseOfConvolutionOdd}). 
The main contribution of this paper, both theoretically and computationally, is
the extension of this recursive inversion framework to broader function spaces
(Theorem \ref{thm:GeneralizationToSchwartzFunctionsAndDistributions}). By
passing the polynomial inversion formula to the limit, we explicitly recover the
original function or tempered distribution as a limiting sequence of linear
combinations of recursive convolutions. This yields a novel, iterative numerical
formula for the inverse of the Weierstrass Transform. We rigorously prove
convergence in the respective topologies ($L^1$, $L^2$, the Schwartz space, and
tempered distributions), ensuring that the method is theoretically correct while
remaining numerically implementable. Finally, in Section \ref{sec:Applications},
we present applications of these results, validating our approach and providing
numerical insight into how this iterative convolution framework translates into
robust signal processing algorithms. We end this work with conclusions in
Section \ref{sec:Conclusions}.

\subsection{Notation}
Before proceeding, it is important to introduce the notation used in this work.
We denote by $\N = \{1,2,\dots\}$ the set of positive integers, and $\N_0 := \N
\cup \{0\}$.  All integrals must be thought with respect to the Lebesgue
measure.  Given a set $X$, we denote by $\id : X \to X$ the identity function
defined as $x\in X \mapsto \id(x) = x$.  The indicator function of a set $A$ is
defined as
\begin{equation*}
    \ind_{A}(x) = \begin{cases} 1 & x \in A, \\ 0 & x \notin A \end{cases}.
\end{equation*}
For a set $A$ we denote its closure as $\overline{A}$. The support of a
real-valued function is defined as
\begin{equation*}
    \supp f := \overline{\{x \in \dom f \mid f(x) \neq 0\}},
\end{equation*}
where $\dom f$ is the domain of $f$. Given $x,y \in \R$
we define $xA + y = \{xa + y \mid x \in A\}$. If $n > 2$ a natural
number, by $\sum_{j \in 2\N_0}^{n}j$ we mean the sum of even nonnegative 
integers up to $n$ if $n$ is even, or up to $n-1$ if $n$ is odd.
The sum $\sum_{j\in2\N}^{n}j$ is then understood as before but for elements in
$2\N$. Given two functions $f,g : \R^d \to \R$ with $d \in \N$ we denote, when
it exists, its convolution as $f * g$. The space of Schwartz functions is
denoted as $\Sw(\R^d)$, while $\SM(\R^d) = \{f \in \Sw(\R^d) \mid \int_{\R^d}f =
1\}$, and $\Sw'(\R^d)$ denotes the space of tempered distributions. For $f \in
\Sw'(\R^d)$ and $\psi \in \Sw(\R^d)$ we denote the image of $\psi$ under $f$
either by $f(\psi)$ or $\dotProduct{f}{\psi}$.  The space of $p \in [1,\infty]$
integrable functions as $L^p(\R^d)$, and the space of locally integrable
functions as $L^p_{loc}(\R^d)$, i.e., $p$-integrable on each compact subset of
$\R^d$. For a function $\varphi : \R^d \to \R$ we will denote by $\varphi_{\ep}
:= \frac{1}{\ep^d}\varphi \circ \frac{\id}{\ep}$ for any $\ep > 0$. We define
the Fourier transform of a function $f : \R \to \R$, when it exists, by
$\mathcal{F}(f)(\xi) = \hat{f}(\xi) := \int_{\R}f(x)e^{-2\pi i\xi x}dx$ for all
$\xi \in \R$, and its inverse Fourier transform by $\mathcal{F}^{-1}(f)$. By
LDCT we refer to the Lebesgue's Dominated Convergence Theorem and by FFT we mean
the Fast Fourier Transform.

\section{Convolution and deconvolution of polynomials and arbitrary 
functions with Schwartz convolution kernels}\label{sec:TheoreticalResults}

Given $\ep > 0$ and even $\varphi \in \Sw(\R)$, it is easy to check that if $p$
is an affine function, then $p * \varphi_{\ep} = p$.  This property does not
hold in general for polynomials $p$ with a degree greater than one, that is, the
fixed points for the (even) convolution operator are constants and affine
functions. However, it is interesting to study if the convolution of a
polynomial results in a polynomial, if this operation can be reversed, and if
the inverse, if it exists, can be extended to more function spaces.

In this section we get the following results for a convolution kernel $\varphi
\in \Sw(\R)$ with $\int_{\R}\varphi \neq 0$ and any polynomial $p$ of degree $n
\in \N_0$ defined in the real numbers.
\begin{itemize}
    \item If $\varphi$ is even, the convolution of $p$ with $\varphi$ results
    in a polynomial of the same degree (Theorem
    \ref{thm:MollificationOfPolynomialsRealLine}). In particular, if the
    original degree of the polynomial is $n\geq 2$, the convolution is the sum
    of the original polynomial with another polynomial of degree $n-2$, one or
    zero, depending on the parity of $n$ (Corollaries
    \ref{cor:CorollaryPolynomials} and \ref{cor:PolynomialsMapping}). In
    particular, the convolution with an even Schwartz function of a polynomial
    is another polynomial with the same asymptotic growth, and in addition the
    first and second leading coefficients are shared between $p$ and its
    convolution.
    \item If $\varphi$ is even, for any $\ep > 0$, the convolution operator $p 
    \mapsto p * \varphi_{\ep}$ is an automorphism between the vector space of
    polynomials of the same degree, for any value of $\ep > 0$ (Theorem
    \ref{thm:ConvolutionIsAnIsomorphismInPolynomials}). This implies that the
    convolution with an even Schwartz kernel in the vector space of polynomials
    is invertible. The inverse can be computed as in Equations
    \eqref{eq:InverseOfConvolution} and \eqref{eq:InverseOfConvolutionOdd}, and
    we provide an insight of its numerical simplicity in Corollary
    \ref{cor:EasyOfNumerical}.  \item We prove in Theorem
    \ref{thm:GeneralizationToSchwartzFunctionsAndDistributions}, that these
    results can be generalized by passing to the limit in Equation
    \eqref{eq:InverseOfConvolution} of Theorem
    \ref{thm:ConvolutionIsAnIsomorphismInPolynomials}, for a certain class of
    functions, and without the need for the Schwartz convolution kernel to be
    even. This results are remarkable, since by an iterative process the
    original function can be restored via convolutions.
    \item Some results can be extrapolated to polynomials in 
    $\R^d$ with $d \geq 2$, see Section \ref{subsec:PolynomialsRd}.
\end{itemize}

Refer to Section \ref{sec:Applications} to several applications of the previous
results.
\subsection{Polynomials in the real line}\label{subsec:PolynomialsRealLine}
Let $p : \R \to \R$ be a polynomial of degree $n \in \N_0$, expressed in the
form
\begin{equation}\label{eq:PolynomialRealLineDefinition}
    p(x) = a_0 + a_1x + \dots + a_nx^n = \sum_{k=0}^{n}a_kx^k, \quad 
    (x \in \R),
\end{equation}
where $\{a_k\}_{k=0}^{n} \subset \R, a_n \neq 0$. The vector space of all
polynomials of degree $n \in \N_0$ will be denoted as $\Poly_n(\R)$. We will say
that $a_{n}$ is the leading coefficient, and $\{a_{n-k}\}_{k=0}^{r}$ are the $r$
leading coefficients for $0 \leq r \leq n$.  For a matter of convenience we will
always exclude the case that $p \equiv 0$, i.e., the zero constant polynomial.

We will show that for any convolution kernel $\varphi \in \SM(\R)$ it holds that
$P_{\ep} := p * \varphi_{\ep}$ is a polynomial of the same degree for $\ep > 0$.
That is, no matter how $\varphi$ is constructed, that polynomials are, in some
sense, invariant under the convolution with Schwartz kernels, i.e.,  the mapping
$p \mapsto P_{\ep}$ maps $\Poly_n(\R)$ into $\Poly_n(\R)$. But we will also
prove that it is an automorphism if $\varphi$ is \textit{even}. But first, we
need a preliminary definition. We are going to restrict ourselves to the set
$\SM(\R)$.  But note that, being the convolution a linear mapping, and given $f
\in \Sw(\R)$, we can define $f / \int_{\R} f \in \SM(\R)$ as long as $\int_{\R}f
\neq 0$. So we restrict ourselves from Schwartz functions with null integral.

\begin{definition}\label{def:PowersOfPolynomialIntegration}
    Let $\varphi \in \SM(\R)$ be even. Define for
    each $m \in \N_0$ the constant
    \begin{equation}\label{eq:ConstantDefinition}
        c_m := \int_{\R}x^m\varphi(x)dx.
    \end{equation}
    It is clear that $c_0 = 1$ and 
    \begin{align*}
        \int_{\R}x^{m}\varphi_{\ep}(x)dx = \begin{cases}
            c_m\ep^m & m \in 2\N_0 \\
            0 & m \in 2\N_0 + 1
        \end{cases}.
    \end{align*}
    Note that the integral is well defined since $\id^m\varphi \in 
    \Sw(\R)$
    for any $m \in \N_0$.
\end{definition}

\begin{remark}
    Note that if $p \in \Poly_n(\R)$, then $\int_{\R}p\varphi_{\ep}$ can
    be considered as a polynomial with even powers in $\ep$ for any even 
    $\varphi \in \SM(\R)$. Indeed if $n > 2$ is even
    \begin{equation*}
        \ep \mapsto h(\ep):=\int_{\R}p(x)\varphi_{\ep}(x)dx = 
        \int_{\R}\sum_{k=0}^{n}a_kx^k\varphi_{\ep}(x)dx = a_0 + a_2c_2\ep^2 
        + \dots + a_nc_n\ep^n,
    \end{equation*}
    while the result is clear if $n$ is odd. Moreover, it is clear
    that if $\varphi \in \Sw(\R)$ is neither even nor odd, we can still
    define $c_m$ in the same way.
\end{remark}

\begin{theorem}\label{thm:MollificationOfPolynomialsRealLine}
    Let $p \in \Poly_n(\R)$, $\varphi \in \SM(\R)$ be 
    even and $\ep >0$.
    
    It holds that 
    \begin{equation}\label{eq:PolyMollificationRealLine}
        P_{\ep}(x) = (p * \varphi_{\ep})(x) = \sum_{j=0}^{n}\sum_{k\in 
            2\N_0}
        ^{j}\binom{j}{k}a_jx^{j-k}c_k\ep^k, \quad (x \in \R).
    \end{equation}
Therefore, the mapping $p \mapsto P_{\ep}$ is a mapping of $\Poly_n(\R)$
into itself.
\end{theorem}
\begin{proof}
    Let $m \in \N$. We know by the Binomial theorem that for all $x 
    \in \R$,
    \begin{align*}
        \int_{\R}(x-y)^{m}\varphi_{\ep}(y)dy &= 
        \int_{\R}\sum_{k=0}^{m}\binom{m}{k}(-1)^{k}x^{m-k}y^k\varphi_{\ep}(y)dy\
        &= 
        \sum_{k=0}^{m}\binom{m}{k}(-1)^{k}x^{m-k}\int_{\R}y^k\varphi_{\ep}(y)dy
        \\
        &= \sum_{k\in 2\N_0}^{m}\binom{m}{k}x^{m-k}c_k\ep^k,
    \end{align*}
    where the last equality comes from Definition 
    \ref{def:PowersOfPolynomialIntegration} and the fact that $(-1)^k = 1$ 
    for even $k$. Now note that since the convolution is a 
    bilinear operation
    \begin{align*}
        P_{\ep}(x) = \sum_{j=0}^{n}a_j\int_{\R}(x-y)^{j}\varphi_{\ep}(y)dy 
        = \sum_{j=0}^{n}a_j\sum_{k\in 
            2\N_0}^{j}\binom{j}{k}x^{j-k}c_k\ep^k, 
        \quad (x \in \R),
    \end{align*}
    and since the sums are finite, we can use Fubini's Theorem to 
    end the proof of Equation \eqref{eq:PolyMollificationRealLine}.
    Note that the right hand side of Equation
    \eqref{eq:PolyMollificationRealLine} is again a polynomial of at most degree
    $n$, proving the theorem.
\end{proof}
It is clear and immediate that if $\varphi$ is not even we still get
a polynomial, the fact that we are restricting to even convolution kernels
will become apparent in Theorem 
\ref{thm:ConvolutionIsAnIsomorphismInPolynomials}.

We now can improve upon the results of Theorem 
\ref{thm:MollificationOfPolynomialsRealLine} using the following corollary.
\begin{corollary}\label{cor:CorollaryPolynomials}
    Let $p \in \Poly_n(\R)$, with $ n \geq 2$ and $\varphi \in 
    \SM(\R)$ be 
    even and let $\ep > 
    0$. It holds that
    \begin{equation}\label{eq:RelationShipPolyMolliPoly}
        P_{\ep}(x) = p(x) + \sum_{j =0}^{n-2}x^{j}\left(\sum_{k\in 
            2\N}^{n-j}\binom{k+j}{k}c_k\ep^ka_{k+j}\right), \quad (x \in \R).
    \end{equation}
    That is, the convolution of a polynomial results in another polynomial
    with the same asymptotic growth, in fact the first two leading coefficients
    of $P_{\ep}$ are the same as those of $p$.
\end{corollary}

\begin{proof}
    We proceed for the even case by induction over $n\geq 2$. Suppose first
    that $n = 2$, then it is straightforward to see using Theorem 
    \ref{thm:MollificationOfPolynomialsRealLine} that 
    \begin{equation*}
        P_{\ep}(x) = a_0 + a_1x + a_2x^2 + a_2c_2\ep^2, \quad (x \in \R),
    \end{equation*}
    and so Equation \eqref{eq:RelationShipPolyMolliPoly} is valid for 
    $n = 2$. Let's suppose it is true for $n \in \N$ and proceed to 
    prove that it is also true for $n+1$. Let $q \in \Poly_{n+1}(\R)$, and 
    denote by $r \in \Poly_{n}(\R)$ the polynomial of degree $n$ obtained 
    by making the leading coefficient of $q$ to be zero. Then
    by Theorem \ref{thm:MollificationOfPolynomialsRealLine},
    \begin{align*}
        Q_{\ep}(x) &:= (q * \varphi_{\ep})(x) = \sum_{j=0}^{n+1}\sum_{k\in 
            2\N_0}
        ^{j}\binom{j}{k}a_jx^{j-k}c_k\ep^k \\
        &= \sum_{j=0}^{n}\sum_{k\in 
            2\N_0}
        ^{j}\binom{j}{k}a_jx^{j-k}c_k\ep^k + \sum_{k\in 
            2\N_0}
        ^{n+1}\binom{n+1}{k}a_{n+1}x^{n+1-k}c_k\ep^k \\
        &= (r * \varphi_{\ep})(x) + \sum_{k\in 
            2\N_0}
        ^{n+1}\binom{n+1}{k}a_{n+1}x^{n+1-k}c_k\ep^k \\
        &= \left[r(x) + \sum_{j=0}^{n-2}x^{j}\left(\sum_{k\in 
            2\N}^{n-j}\binom{k+j}{k}c_k\ep^ka_{k+j}\right) \right]
        +\sum_{k\in 
            2\N_0}
        ^{n+1}\binom{n+1}{k}a_{n+1}x^{n+1-k}c_k\ep^k \\
        &= q(x) + \sum_{j =0}^{n-2}x^{j}\left(\sum_{k\in 
            2\N}^{n-j}\binom{k+j}{k}c_k\ep^ka_{k+j}\right) 
        +\sum_{k\in 
            2\N}
        ^{n+1}\binom{n+1}{k}a_{n+1}x^{n+1-k}c_k\ep^k.
    \end{align*}
    Suppose that $n$ is even, then
    \begin{align*}
        \sum_{k\in 
            2\N}
        ^{n+1}\binom{n+1}{k}a_{n+1}x^{n+1-k}c_k\ep^k
        = x&^{n-1}\binom{n+1}{2}c_2\ep^2 a_{n+1} + 
        x^{n-3}\binom{n+1}{4}c_4\ep^4 a_{n+1} + \dots + 
        x\binom{n+1}{n}c_n\ep^n a_{n+1},
    \end{align*}
    which implies
    \begin{align*}
        &\sum_{j =0}^{n-2}x^{j}\left(\sum_{k\in 
            2\N}^{n-j}\binom{k+j}{k}c_k\ep^ka_{k+j}\right) 
        + 
        \sum_{k\in 
            2\N}
        ^{n+1}\binom{n+1}{k}a_jx^{n+1-k}c_k\ep^k 
        = 
        \sum_{j=0}^{n-1}x^j\left(\sum_{k\in 
            2\N}^{n+1-j}\binom{k+j}{k}c_k\ep^{k}a_{k+j}\right),
    \end{align*}
    completing the proof for even $n$. If $n$ is odd the proof 
    is similar.
\end{proof}
\begin{remark}\label{rmk:NotationSidePolynomial}
    Given $p \in \Poly_{n}(\R)$ and $n \geq 2$. We will denote by
    \begin{equation*}
        b_{j} := \sum_{k\in 
            2\N}^{n-j}\binom{k+j}{k}c_k\ep^ka_{k+j}, \quad j \in 
        \{0,1,\dots,n-2\},
    \end{equation*}
    In doing so Equation \eqref{eq:RelationShipPolyMolliPoly} can be written as 
    \begin{equation*}
        P_{\ep}(x) = p(x) + \pp(x), \quad (x \in \R),
    \end{equation*}
    where $\pp \in \Poly_{n-2}(\R)$ with coefficients $\{b_j\}_{j=0}^{n-2}$,
    i.e., $\pp(x) = b_0 + b_1x + \dots + b_{n-2}x^{n-2}$ for all $x \in 
    \R$. Note that $b_{n-2} \neq 0$, since $b_{n-2} = c_2\ep^2a_n\binom{n}{2}$
    and $a_n \neq 0, c_2 \neq 0$. 
    This leads to the 
    following corollary.
\end{remark}
\begin{corollary}\label{cor:PolynomialsMapping}
    Let $p \in \Poly_n(\R)$ with $n \geq 2$, $\varphi \in \SM(\R)$ be even and
    $\ep > 0$. Define by $T_{\ep} : \Poly_n(\R) \to \Poly_n(\R)$  the map
    \begin{equation*}
        p \mapsto T_{\ep}(p) = p * \varphi_{\ep}.
    \end{equation*}
    Given $k \in \N_0$ denote the composition of $T_{\ep}$ with itself $k$
    times as $T_{\ep}^k$, i.e., the convolution applied recursively $k$ 
    times, with $T_{\ep}^{0} := \id$. The following equation is true
    \begin{equation}\label{eq:CompositionFolding}
        T_{\ep}^k(p) = \sum_{j=0}^{k}\binom{k}{j}p_j,
    \end{equation}
    where $p_j$ is defined recursively as follows. If $j = 0$ then
    $p_0 := p$. If $j = 1$ then $p_j := \pp$ as in Remark 
    \ref{rmk:NotationSidePolynomial}, if $j = 2$ then $p_j := \tilde{\pp}$,
    which comes from $T(\pp) = \pp + \tilde{\pp}$, as in Remark
    \ref{rmk:NotationSidePolynomial}, and so on.
    Moreover the following statements are also true.
    \begin{enumerate}
        \item Suppose $n$ is even. If $k \in \N$ satisfies that 
        $2k \leq n$ then
        $p_j \in \Poly_{n-2j}(\R)$ for $j \in \{0,\dots,k\}$. If $2k > n$ 
        then, $p_j \in \Poly_{n-2j}(\R)$ for 
        $j \in \{0,\dots,\frac{n}{2}-1\}$ and $p_j \in \Poly_0(\R)$ for $j 
        \in \{\frac{n}{2},\dots,k\}$, that is the last $k-\frac{n}{2}$ 
        polynomials are 
        constants.
        \item Suppose $n$ is odd. If $k \in \N$ satisfies that 
        $2k < n$ then
        $p_j \in \Poly_{n-2j}(\R)$ for $j \in \{0,\dots,k\} \subset 
        \{0,\dots,\frac{n-1}{2}\}$. 
        If $2k > n$ then $p_j \in \Poly_{n-2j}(\R)$ for $j \in 
        \{0,\dots,\frac{n-1}{2}-1\}$ and $p_j \in \Poly_{1}(\R)$
        for $j \in \{\frac{n-1}{2},\dots,k\}$, that is, the last 
        $k-\frac{n-1}{2}$ polynomials are affine functions.
    \end{enumerate}
\end{corollary}
\begin{proof}
    We first proceed with the proof of Equation 
    \eqref{eq:CompositionFolding}.
    Fix $n \in \N$ with $n \geq 2$. We proceed by induction over $k$. For $k
    = 1$ it holds that 
    \begin{equation*}
        T_{\ep}(p) = P_{\ep} = p + \pp = p + p_1,
    \end{equation*}
    where $p_1 \in \Poly_{n-2}(\R)$ by Corollary 
    \ref{cor:CorollaryPolynomials}. Note that $T$ is linear, which implies
    \begin{equation*}
        T_{\ep}(T_{\ep}(p)) = T_{\ep}(p+p_1) = T_{\ep}(p) + 
        T_{\ep}(p_1) = p 
        + p_1 + p_1 + 
        p_2 = p + 2p_1 + p_2.
    \end{equation*}
    So the statement is proven for $k = 2$, because $p_1 \in 
    \Poly_{n-2}(\R)$ by Corollary \ref{cor:CorollaryPolynomials}
    and Remark \ref{rmk:NotationSidePolynomial}, which by the same
    conclusions imply that $p_2 \in \Poly_{\max\{n-4,0\}}(\R)$.
    
    Suppose the statement is true for $k \in \N$, and $2k \leq n$, that is
    $T_{\ep}^k(p) = \sum_{j=0}^{k}\binom{k}{j}p_j$, with $p_j \in
    \Poly_{n-2j}(\R)$ with $j \in \{0,\dots,k\}$. Note $T_{\ep}$ is linear by
    the properties of the convolution, thus
    \begin{equation*}
        T_{\ep}^{k+1}(p) = T_{\ep}(T_{\ep}^k(p)) = 
        T_{\ep}\left(\sum_{j=0}^{k}\binom{k}{j}p_j\right) = 
        \sum_{j=0}^{k}\binom{k}{j}T_{\ep}(p_j) = 
        \sum_{j=0}^{k}\binom{k}{j}(p_j + p_{j+1}),
    \end{equation*}
    and then
    \begin{align*}
        T_{\ep}^{k+1}(p) &= 
        \sum_{j=1}^{k}\left(\binom{k}{j}+\binom{k}{j-1}\right)p_j + p_0 
        + p_{k+1} 
        = \sum_{j=1}^{k}\binom{k+1}{j}p_j + p_0 + p_{k+1} 
        = \sum_{j=0}^{k+1}\binom{k+1}{j}p_j,
    \end{align*}
    where we used the fact that $\binom{k}{j}+\binom{k}{j-1}=\binom{k+1}{j}$.
    Moreover, $p_{k} \in \Poly_{n-2k}(\R)$, which implies by 
    Corollary \ref{cor:CorollaryPolynomials} and Remark
    \ref{rmk:NotationSidePolynomial} that $p_{k+1} \in 
    \Poly_{n-2(k+1)}(\R)$. This proves Equation 
    \eqref{eq:CompositionFolding}. Note that the formula
    is always valid, independently if $n$ is even or odd,
    so what we need now is to prove it for
    when $2k > n$.
    
    \begin{enumerate}
        \item Suppose $n$ is even and $2k \leq n$, then the previous 
        discussion is valid, and so $p_j \in \Poly_{n-2j}(\R)$
        for $j \in \{0,\dots,k\}$. Suppose that $2k > n$, then
        $k > \frac{n}{2}$ and for $j = n/2$ it holds by the 
        previous proof that $p_{j} \in \Poly_{0}(\R)$, i.e.,
        it is a constant. Since constant functions are invariant 
        under convolution with $\varphi \in \SM(\R)$, it holds that 
        $T_{\ep}(p_j) = p_j$ for 
        $j \in \{\frac{n}{2},\dots,k\}$, and this proves
        the first assertion.
        \item Suppose $n$ is odd and $2k < n$, then $2k \leq 
        n-1$, so the previous discussion is valid, that is
        $p_j \in \Poly_{n-2j}(\R)$ for $j \in \{0,\dots,k\}$.
        Now, if $2k > n$ then $2k \geq n+1$, and for 
        $j = \frac{n-1}{2}$ it holds that $p_j \in \Poly_{n-2j}(\R)
        = \Poly_{1}(\R)$, but then $p_j$ is affine, which implies
        that $T_{\ep}(p_j) = p_j$ for $j = \frac{n-1}{2}$ and so for 
        $j \in \{\frac{n-1}{2},\dots,k\}$, proving the second assertion.
    \end{enumerate}
\end{proof}

It will be useful to compute the polynomials
$p_j$ in Corollary \ref{cor:PolynomialsMapping} as linear combinations of 
repeated convolutions of $p_0 = p$, i.e., of the original polynomial. This will
allow to obtain a formula for the inversion of the convolution on the space
$\Poly_n(\R)$. We prove that this is possible in the next lemma.
\begin{lemma}\label{lem:ComputationOfSidePolynomials}
    Let $p \in \Poly_{n}(\R)$ with $n \geq 2$, $\varphi \in \SM(\R)$
    be even and $\ep > 0$. Defining $T_{\ep} : \Poly_{n}(\R) \to 
    \Poly_{n}(\R)$
    as in Corollary \ref{cor:PolynomialsMapping} as well as $p_j$
    for $j \in \N_0$, then
    \begin{equation}\label{eq:SidePolynomials}
        p_j = \sum_{k=0}^{j}\binom{j}{k}(-1)^{k}T_{\ep}^{j-k}(p),
    \end{equation}
    where $T_{\ep}^0 := \id$ for all $\ep > 0$  and $p_0 := p$.
\end{lemma}
\begin{proof}
    We proceed by induction over $j$. By definition it is true for $j=0$, 
    and for $j = 1$ we know by Theorem 
    \ref{thm:MollificationOfPolynomialsRealLine} and Corollary 
    \ref{cor:PolynomialsMapping} that $T_{\ep}(p_0) = p_0 + p_1$ and 
    therefore
    $p_1 = T_{\ep}(p_0) - p_0$ for which Equation
    \eqref{eq:SidePolynomials}
    clearly holds. Suppose now that it is true for $j \in \N$. Then 
    $T_{\ep}(p_j) = p_j + p_{j+1}$ and so
    \begin{align*}
        p_{j+1} = T_{\ep}(p_j) - p_{j} &= 
        \sum_{k=0}^{j}\binom{j}{k}(-1)^{k}T_{\ep}^{j+1-k}(p_0) + 
        \sum_{k=0}^{j}\binom{j}{k}(-1)^{k+1}T_{\ep}^{j-k}(p_0) \\
        &=\sum_{k=0}^{j}\binom{j}{k}(-1)^{k}T_{\ep}^{j+1-k}(p_0) +
        \sum_{k=1}^{j+1}\binom{j}{k-1}(-1)^{k}T_{\ep}^{j+1-k}(p_0) \\
        &= T_{\ep}^{j+1}(p_0) + (-1)^{j+1}T_{\ep}(p_0) +
        \sum_{k=1}^{j}\binom{j+1}{k}(-1)^{k}T_{\ep}^{j+1-k}(p_0) \\
        &= 
        \sum_{k=0}^{j+1}\binom{j+1}{k}(-1)^{k}T_{\ep}^{j+1-k}(p_0),
    \end{align*}
    where the fact $\binom{j+1}{k} = \binom{j}{k-1} + \binom{j}{k}$ has been
    used in the second to last inequality. This proves Equation 
    \eqref{eq:SidePolynomials}, thus the lemma.
\end{proof}

We now prove that the mapping $T_{\ep}$, defined in Corollary
\ref{cor:PolynomialsMapping}, is an isomorphism on the vector space
$\Poly_{n}(\R)$. Rather than exploiting the finite-dimensional nature of
$\Poly_n(\R)$ to construct the matrix for the linear isomorphism $T_{\ep}$, our
proof proceeds recursively by using the results established in
Corollary \ref{cor:CorollaryPolynomials} and Remark
\ref{rmk:NotationSidePolynomial}. We adopt this recursive approach because it
serves as the foundational mechanism that will allow us to extend the inversion
to infinite-dimensional spaces of arbitrary functions and distributions in
Theorem \ref{thm:GeneralizationToSchwartzFunctionsAndDistributions}.
Nevertheless, this isomorphism establishes a remarkable structural property: within
$\Poly_{n}(\R)$, convolution with an even Schwartz kernel is always an exactly
reversible operation. Consequently, this framework can be viewed as an extension
of the Weierstrass transform applied to polynomials, where $\varphi \in \SM(\R)$
represents the even Gaussian kernel (see \cite{bilodeau1962weierstrass} for a
related discussion on Hermite polynomials and the Weierstrass function). 
Finally, we emphasize that the parity of the convolution kernel plays
a critical role in the following proof. Its necessity is twofold: first, for
any $p \in \Poly_{n}(\R)$, the operator yields $T_{\ep}(p) = p + p_1$, where
$p_1 \in \Poly_{n-2}(\R)$. Second, it ensures that affine and constant functions
remain completely invariant (i.e., $T_{\ep}(p) = p$). This degree reduction is
precisely what enables both the proof of the automorphism and the explicit,
finite-sum computation of its inverse.

\begin{theorem}\label{thm:ConvolutionIsAnIsomorphismInPolynomials}
    Let $n \in \N$, $\varphi \in \SM(\R)$ be even and 
    $\ep > 0$.
    Then $T_{\ep} : \Poly_{n}(\R) \to \Poly_{n}(\R)$, defined as in 
    Corollary \ref{cor:PolynomialsMapping} is an automorphism. 
    Moreover, given $p \in \Poly_{n}(\R)$ with $n$ even
    \begin{equation}\label{eq:InverseOfConvolution}
        T_{\ep}^{-1}(p) = \sum_{j=0}^{n/2}(-1)^{j}p_j =
        \sum_{j=0}^{n/2}\sum_{k=0}^{j}(-1)^{j+k}\binom{j}{k}T_{\ep}^{j-k}(p) =
        \sum_{j=0}^{n/2}(-1)^{n/2-j}\binom{n/2+1}{j}T_{\ep}^{n/2-j}(p),
    \end{equation}
    where $p_0 := p$, and $T_{\ep}^{0} := \id$,
    while if $n$ is odd
    \begin{equation}\label{eq:InverseOfConvolutionOdd}
        T_{\ep}^{-1}(p) = \sum_{j=0}^{(n-1)/2}(-1)^{j}p_j = 
        \sum_{j=0}^{(n-1)/2}(-1)^{\frac{n-1}{2}-j}\binom{\frac{n+1}{2}}{j}T_{\ep}^{\frac{n-1}{2}-j}(p).
    \end{equation}
\end{theorem}
\begin{proof}
    The fact that $T_{\ep}$ is linear follows directly from the definition
    of the convolution. We proceed to prove that it is bijective. Suppose 
    first that $n = 2$. We know by
    Corollary \ref{cor:PolynomialsMapping} that $T_{\ep}(p) = p + p_1$
    where $p_1 \in \Poly_{0}(\R)$, and therefore $T_{\ep}(p_1) = p_1$.
    If $p,q \in \Poly_{2}(\R)$ such that $p(x) = a_0+a_1x+a_2x^2$ and
    $q(x) = b_0+b_1x+b_2x^2$ for all $x \in \R$, then
    $p_1(x) = a_2\ep^2c_2$ and $q_1(x) = b_2\ep^2c_2$. Thus
    \begin{equation*}
        T_{\ep}(p) = T_{\ep}(q) \iff p + p_1 = q + q_1,
    \end{equation*}
    but since $p+p_1 \in \Poly_{2}(\R)$ and $q+q_1 \in \Poly_2(\R)$
    with the same leading coefficients as $p$ and $q$ by Corollary 
    \ref{cor:CorollaryPolynomials},
    and two polynomials of the same degree are equal if and only if all of 
    its coefficients
    are equal, it follows that 
    \begin{equation*}
        b_2 = a_2, b_1 = a_1, b_0+b_2\ep^2c_2 = a_0+a_2\ep^2c_2,
    \end{equation*}
    from which it follows that $b_0 = a_0$, that is $p = q$, proving
    that $T_{\ep}$ is injective. Clearly given $p \in \Poly_{2}(\R)$
    the polynomial $p-p_1$ satisfies that $T_{\ep}(p-p_1) = 
    T_{\ep}(p)-T_{\ep}(p_1) = p$
    proving that $T_{\ep}$ is an isomorphism of $\Poly_{2}(\R)$ onto itself.
    Therefore, we can know generalize this concepts by induction. 
    \begin{itemize}
        \item $T_{\ep}$ is injective. Let $n \in \N$. Suppose $T(p) = 
        T(q)$, then 
        $p+p_1 = q+q_1$. Since $p_1,q_1 \in \Poly_{n-2}(\R)$ this implies
        that $a_n = b_n$ and $a_{n-1}=b_{n-1}$. Now, note that by Remark
        \ref{rmk:NotationSidePolynomial} that
        \begin{equation*}
            a_{n-2} + \binom{n}{2}c_2\ep^2a_{n} = 
            b_{n-2}+\binom{n}{2}c_2\ep^2b_{n},
        \end{equation*}
        and therefore $a_{n-2} = b_{n-2}$. We can continue in this manner
        to prove that $a_{j} = b_{j}$ for all $j \in \{0,\dots,n\}$,
        that is, that $p = q$, proving the injection.
        
        \item $T_{\ep}$ is surjective. Suppose $n$ is even. If $n=2$, by the
        previous proof, for all $p \in \Poly_{2}(\R)$ there exists $q_{p} \in
        \Poly_{2}(\R)$ such that $T_{\ep}(q_{p}) = p$. Now let $p \in
        \Poly_{4}(\R)$, we know that $T_{\ep}(p) = p +p_1$ with $p_1 \in
        \Poly_{2}(\R)$, and so there exists $q_{p_1} \in \Poly_{2}(\R)$ such
        that $T_{\ep}(q_{p_1}) = p_1$. Therefore $T_{\ep}(p-q_{p_1}) =
        T_{\ep}(p)-T_{\ep}(q_{p_1}) = p + p_1 - p_1 = p$, and so $T_{\ep}$ is a
        surjection of $\Poly_{4}(\R)$ onto itself. We can now easily proceed by
        induction to prove it is a surjection for $n \in 2\N$. 
        
        Suppose now that $n = 3$. Then $T(p) = p+p_1$ with $p \in 
        \Poly_{1}(\R)$, which implies that $T(p_1) = p_1$. Therefore,
        $q = p-p_1$ satisfies $T_{\ep}(q) = T_{\ep}(p)-T_{\ep}(p_1) = p$
        proving the surjection for $n = 3$. Suppose it is true for $3\leq n
        \in 2\N +1 $. Given $p \in \Poly_{n+2}(\R)$ we know that
        $T_{\ep}(p) = p +p_1$ with $p_1 \in \Poly_{n}(\R)$, and so
        by the induction hypothesis there exists $q_1 \in \Poly_{n}(\R)$
        such that $T_{\ep}(q_1) = p_1$. Therefore, choosing 
        $q = p-q_1$ the surjection is proven.
    \end{itemize}
    We proceed now to prove Equation \eqref{eq:InverseOfConvolution}.
    For $n=2$ the formula reads $T_{\ep}^{-1}(p) = p -p_1$ which is true,
    since $T_{\ep}(T_{\ep}^{-1}(p)) = T_{\ep}(p)-T_{\ep}(p_1) = p+p_1-p_1 = 
    p$. Suppose then 
    it is true for $n \geq 2$. Given $p \in \Poly_{n+2}(\R)$ we know 
    that $T_{\ep}(p) = p + p_1$ where $p_1 \in \Poly_{n}(\R)$. If
    $q = p-T_{\ep}^{-1}(p_1)$, then $T_{\ep}(q) = T_{\ep}(p)-p_1 = 
    p+p_1-p_1 = p$. That is,
    \begin{equation*}
        q = T_{\ep}^{-1}(p) = p-T_{\ep}^{-1}(p_1) = p - 
        \sum_{j=0}^{n/2}(-1)^{j}p_{j+1} = p_0 + 
        \sum_{j=1}^{n/2+1}(-1)^{j}p_{j} = \sum_{j=0}^{(n+2)/2}(-1)^{j}p_{j},
    \end{equation*}
    proving the first equality of Equation \eqref{eq:InverseOfConvolution}.
    The second equality comes from Lemma 
    \ref{lem:ComputationOfSidePolynomials}. Finally, the third 
    equality can we proven by induction. For $n=2$ we have that 
    $T_{\ep}^{-1}(p) = 2p - T_{\ep}(p)$ for both formulas. Suppose then
    it holds for $n \geq 2$ with $n$ even. Then for $n+2$
    \begin{align*}
        T_{\ep}^{-1}(p) &= 
        \sum_{j=0}^{(n+2)/2}\sum_{k=0}^{j}(-1)^{j+k}\binom{j}{k}T_{\ep}^{j-k}(p)
        \\
        &=\sum_{j=0}^{n/2}\sum_{k=0}^{j}(-1)^{j+k}\binom{j}{k}T_{\ep}^{j-k}(p)
        +\sum_{k=0}^{n/2+1}(-1)^{n/2+1+k}\binom{n/2+1}{k}T_{\ep}^{n/2+1-k}(p)
        \\
        &= 
        \sum_{j=0}^{n/2}(-1)^{\frac{n}{2}-j}\binom{\frac{n}{2}+1}{j}T_{\ep}^{\frac{n}{2}-j}(p)
        +\sum_{k=1}^{\frac{n}{2}+1}(-1)^{\frac{n}{2}+1-k}\binom{\frac{n}{2}+1}{k}T_{\ep}^{\frac{n}{2}+1-k}(p)
        +(-1)^{\frac{n}{2}+1}T_{\ep}^{\frac{n}{2}+1}(p) \\
        &= 
        \sum_{j=0}^{n/2}(-1)^{n/2-j}\left[\binom{\frac{n}{2}+1}{j}+\binom{\frac{n}{2}+1}{j+1}\right]T_{\ep}^{\frac{n}{2}-j}(p)
        +(-1)^{\frac{n}{2}+1}\binom{\frac{n+2}{2}+1}{0}T_{\ep}^{\frac{n}{2}+1}(p)\\
        &= 
        \sum_{j=0}^{(n+2)/2}(-1)^{\frac{n+2}{2}-j}\binom{\frac{n+2}{2}+1}{j}T_{\ep}^{\frac{n+2}{2}-j}(p).
    \end{align*}
    proving Equation \eqref{eq:InverseOfConvolution}.
    The proof for the case where 
    $n$ odd,
    i.e., Equation \eqref{eq:InverseOfConvolutionOdd},
    is identical to the even case.
\end{proof}
\begin{corollary}\label{cor:EasyOfNumerical}
    Equation \eqref{eq:InverseOfConvolution} can also be computed as
    \begin{equation*}
        T_{\ep}^{-1} = 
        \sum_{j=0}^{n/2}(-1)^{j}\binom{n/2+1}{j+1}T_{\ep}^{j}
    \end{equation*}
    which in turn implies that in $\Poly_{n}(\R)$, with even $n$,
    \begin{equation*}
        T_{\ep}^{-1}\circ T_{\ep} = \id = 
        -\sum_{j=1}^{n/2+1}\binom{n/2+1}{j}(-T_{\ep})^{j} = \id - (\id - 
        T_{\ep})^{n/2+1},
    \end{equation*}
    that is $(\id - T_{\ep})^{n/2+1} \equiv 0$ on $\Poly_n(\R)$.
\end{corollary}
\begin{proof}
    The first result follows from a change of indexes in the
    sum. The second follows noting that, if $f : V \to V$ is a linear 
    mapping, where $V$ is any vector space,
    then, by the binomial theorem for linear functions,
    $\sum_{k=0}^{n}(-1)^{k}\binom{n}{k}f^k = (\id -f)^{n}$ where $\id : V \to V$
    is the identity on $V$.
\end{proof}
Note from Corollary \ref{cor:EasyOfNumerical} it is only needed
to compute $n/2$ convolutions numerically speaking, because each time
the index of the sum advances, it is only needed to convolve the 
previously convolved function.
\begin{remark}
    Theorem \ref{thm:ConvolutionIsAnIsomorphismInPolynomials} states that, for
    polynomials, convolution is invertible as long as the kernel is an even
    (Schwartz) function. Note that convolution is, practically speaking, the
    filtering of the function with a smooth filter in $\SM(\R)$. Note that since
    $\SM(\R) \subset \Sw(\R) \subset L^2(\R)$ the filter has finite energy. What
    the Theorem states is that, even after having (for example) low-filtered a
    polynomial via convolution with an even convolution kernel (e.g., Gaussian
    Kernel, or the compactly supported convolution kernel $x \mapsto \varphi(x)
    = \ind_{(-1,1)}\exp(-1/(1-x^2))$), we can recover the original function
    uniquely.  This powerful fact can also be numerically implemented easily by
    Equations \eqref{eq:InverseOfConvolution} and
    \eqref{eq:InverseOfConvolutionOdd}. See Section \ref{sec:Applications} for
    an application related to this remark.
\end{remark}
\begin{remark}\label{rmk:PolynomialDecomposition}
    Suppose $n$ is even. Note that $T_{\ep}(T^{-1}_{\ep}(p)) = p$ and $T_{\ep}$
    is linear, which implies
    \begin{equation*}
        p =\sum_{j=0}^{n/2}(-1)^{n/2-j}\binom{n/2+1}{j}T_{\ep}^{n/2+1-j}(p),
    \end{equation*}
    that is, any polynomial can be expressed as a linear combination of its
    convolutions with even Schwartz convolution  kernels, no matter the $\ep>0$
    nor the (even) convolution kernel. This is an interesting result: the sum of
    the convolved versions of the polynomial generate the same polynomial. For
    example, for $n=4$,
    \begin{equation*}
        p = T_{\ep}^{3}(p)-3T_{\ep}^2(p) + 3T_{\ep}(p) = p * 
        (3\varphi_{\ep} - 3\varphi_{\ep}*\varphi_{\ep} + 
        \varphi_{\ep}*\varphi_{\ep}*\varphi_{\ep}).
    \end{equation*}
    An in general, for even $n$ we have that
    $\sum_{j=0}^{n/2}(-1)^{n/2-j}\binom{n/2+1}{j}T_{\ep}^{n/2-j}(\varphi_{\ep})$
    is an identity operator for the convolution of polynomials in
    $\Poly_{n}(\R)$ with even Schwartz kernels. Clearly, the identity operator
    is not a polynomial, and thus they do not form a group.  Moreover, note that
    being $T_{\ep}$ an automorphism in $\Poly_{n}(\R)$ then it cannot hold that
    $T_{\ep}(p)$ is identically zero for any $p \in \Poly_{n}(\R)$ except for $p
    \equiv 0$. This is congruent with information theory, we can \textbf{always}
    recover the original polynomial after convolution because there is
    \textbf{always} information about it after the convolution, i.e.,
    $T_{\ep}(p) \not \equiv 0$ for any $p \not\equiv 0$. 
\end{remark}

We end this section by noting the following. There is nothing special
about $\varphi$ being a Schwartz function, since as it can be seen
in the proofs of the previous results, the rapidly decreasing property
of $\varphi$ is \textit{only} needed so that the convolution is well posed.
Therefore, any other even function for which $p * \varphi$, as
well as its recursive convolutions, make sense, is equally valid. Nevertheless,
in the following section we will show that the requirement for $\varphi$
to be a Schwartz function is essential to carry the generalization to more 
function spaces.

\subsection{The extension to a bigger space of functions: the deconvolution
    of arbitrary functions}
The natural question now is to ask if, for a given function $f$, it holds (under
certain assumptions about $f$) that, by passing to the limit in Equation
\eqref{eq:InverseOfConvolution}, we can recover $f$ from $T_{\ep}(f)$. That is,
if $T_{\ep,n}^{-1}$ is the inverse of the automorphism $T_{\ep}$ in
$\Poly_{n}(\R)$ for some even $\varphi \in \Sw(\R)$ (see Theorem
\ref{thm:ConvolutionIsAnIsomorphismInPolynomials}) does it hold that
\begin{equation*}
    \lim_{n\to\infty}T_{\ep,n}^{-1}(T_{\ep}(f)) = f,
\end{equation*}
in some sense of convergence? The answer is affirmative under some
conditions on the Fourier transform of $\varphi \in \Sw$ and different
assumptions for $f$. We prove these results in the next theorem, in which the
assumption that $\varphi$ is even is not always required. As mention in the
introduction, we will denote $\Fourier{f}$ or $\widehat{f}$ the Fourier
transform of a function $f$ whenever it exists, and as $\mathcal{F}$ the Fourier
transform as an operator. For its properties  refer to
\cite{grubb2008distributions} and \cite{grafakos2008classical}.
\begin{theorem}\label{thm:GeneralizationToSchwartzFunctionsAndDistributions}
    Let $\varphi \in \Sw(\R)$ and let $\ep > 0$. Suppose that
    for all $\xi \in \R$ it holds that $0 < \widehat{\varphi_{\ep}}(\xi)  
    < 
    2$. Let the operator $f \mapsto T_{\ep}(f) = f * \varphi_{\ep}$, when 
    it makes sense for 
    $f$,
    and for  $n \in \N$ let $T_{\ep,n}^{-1}$
    be the operator
    \begin{equation}\label{eq:InverseOfConvolutionArbitrary}
        f \mapsto T_{\ep,n}^{-1}(f) = 
        \sum_{k=0}^{n}(-1)^{k}\binom{n+1}{k+1}T_{\ep}^{k}(f),
    \end{equation}
    when it makes sense for $f$, where $T_{\ep}^{k}$ is the $k$-th folding
    composition of $T_{\ep}$. That is, the convolution applied repetitively $k$
    times. If $f \in \Sw'(\R)$ denote by $f * \varphi_{\ep}$ the tempered
    distribution $\dotProduct{f*\varphi_{\ep}}{\psi} =
    \dotProduct{f}{\check{\varphi_{\ep}}*\psi}$ for all $\psi \in \Sw(\R)$,
    where $\check{\varphi_{\ep}}(x) = \varphi_{\ep}(-x)$ for all $x \in \R$.
    Extend the definition of the operators $T_{\ep}$ and $T_{\ep,n}^{-1}$ to
    $\Sw'$. Then the following statements are true
    \begin{enumerate}
        \item If $f \in L^2(\R)$ then 
        $\lim_{n\to\infty}T_{\ep,n}^{-1}(T_{\ep}(f)) = f$ in 
        $L^2(\R)$.
        \item  If $f \in L^1(\R)$ with $\widehat{f} \in L^1(\R)$ then 
        $\lim_{n\to\infty}T_{\ep,n}^{-1}(T_{\ep}(f)) = f$ 
        pointwise.
        \item Let $\one \in \Sw'$ be the one tempered distribution (i.e., 
        $\dotProduct{\one}{\psi} = \int_{\R}\psi$ for all $\psi \in 
        \Sw(\R))$, let $\delta \in \Sw'$ be the Dirac delta distribution,
        let $\vp \in \Sw'$ be the Cauchy Principal Value distribution,
        and for $y \in \R$ 
        let $e^{2\pi iy\id }$ the distribution that arises from the 
        (slowly increasing) function $x \mapsto e^{2\pi ixy}$.  Then,
        for any of these distributions $u \in \Sw'$ it holds that 
        $T_{\ep,n}^{-1}(T_{\ep}(u)) \to u$ as $n \to \infty$ in the 
        distributional topology. The same holds for any linear combination
        or/and translations of these distributions.
        \item Suppose now that $\varphi$ is the Gaussian Kernel, i.e., 
        $\varphi(x) = \frac{1}{\sqrt{4\pi}}e^{-(x/2)^2}$ for all $x \in \R$.
        Then, given $f \in \Sw(\R)$ it holds that 
        $T_{\ep,n}^{-1}(T_{\ep}(f)) \to f$ as $n \to \infty$ in the 
        Schwartz topology, and
        for any $f \in \Sw'(\R)$, it holds $T_{\ep,n}^{-1}(T_{\ep}(f)) \to 
        f$
        as $n \to \infty$ in the tempered distribution sense.
    \end{enumerate}
\end{theorem}
\begin{proof}
    We prove each statement separately.
    \begin{enumerate}
        \item First note that since $f \in L^2(\R)$ and $\varphi_{\ep} \in 
        \Sw(\R)$ then $T_{\ep}(f) = f * \varphi_{\ep} \in L^2(\R)$ by Young's
        inequality for convolution \cite{folland1999real}, which implies that
        $T_{\ep,n}^{-1}(T_{\ep}(f)) \in L^2(\R)$ and so the operator makes
        sense. By the Fourier transform properties, it holds that
        $\Fourier{T_{\ep,n}^{-1}(T_{\ep}(f))} \in L^2(\R)$ for all $n \in \N$
        and all $f \in L^2(\R)$. Moreover $\widehat{\varphi_{\ep}} \in \Sw(\R)$,
        and since $\Fourier{T_{\ep}(f)} = \widehat{f} \widehat{\varphi_{\ep}}$,
        by induction it is easily seen that $\Fourier{T^k_{\ep}(f)} =
        \widehat{f}(\widehat{\varphi_{\ep}}\cdots\widehat{\varphi_{\ep}}) =
        \widehat{f}\widehat{\varphi_{\ep}}^k$ for all $k \in \N$. Thus, by the
        linearity of the Fourier transform
        \begin{align*}
            \Fourier{T_{\ep,n}^{-1}(T_{\ep}(f))} =
            \sum_{k=0}^{n}(-1)^{k}\binom{n+1}{k+1}\widehat{f}\widehat{\varphi_{\ep}}^{k+1}
            = \hat f
            \sum_{k=0}^{n}(-1)^{k}\binom{n+1}{k+1}\widehat{\varphi_{\ep}}^{k+1}.
        \end{align*}
        Moreover,
        \begin{align*}
            \sum_{k=0}^{n}(-1)^{k}\binom{n+1}{k+1}\widehat{\varphi_{\ep}}^{k+1}
            =
            \sum_{k=1}^{n+1}(-1)^{k-1}\binom{n+1}{k}\widehat{\varphi_{\ep}}^{k}
            = -
            \sum_{k=1}^{n+1}\binom{n+1}{k}\left(-\widehat{\varphi_{\ep}}\right)^{k},
        \end{align*}
        but by the binomial theorem $(1-x)^{n+1} = 
        \sum_{k=0}^{n+1}\binom{n+1}{k}(-x)^{k}$, and therefore
        \begin{align*}
            -\sum_{k=1}^{n+1}\binom{n+1}{k}\left(-\widehat{\varphi_{\ep}}\right)^{k}
            = 1-(1-\widehat{\varphi_{\ep}})^{n+1}.
        \end{align*}
        Since $\widehat{\varphi_{\ep}} \in \Sw$ then so it does
        $\widehat{\varphi_{\ep}}^{k}$ for all $k \in \N$, and so
        $1-(1-\widehat{\varphi})^{n+1} \in \Sw$ for all $n \in \N$, being $\Sw$
        a vector space. Now we claim that
        $\widehat{f}(1-(1-\widehat{\varphi_{\ep}})^n) \to \widehat{f}$ in
        $L^2(\R)$. Indeed, note that $|\widehat{f}|^2 \in L^1(\R)$, and also
        \begin{equation*}
            |\widehat{f}(1-(1-\widehat{\varphi_{\ep}})^n) - \widehat{f}|^2 = 
            |\widehat{f}|^2|1-\widehat{\varphi_{\ep}}|^{2n} < 
            |\widehat{f}|^2, \quad 
            \forall n \in \N,
        \end{equation*}
        due to the fact that $ 0 < \widehat{\varphi_{\ep}} < 2$. It also holds
        that\footnote{If there exists some $x \in \R$ in which
            $\widehat{\varphi_{\ep}}(x)  =0 $, then we cannot conclude pointwise
            convergence. This subtle assumption has important consequences.
            Indeed, if the frequency content has been removed, it cannot be
            recovered uniquely, since there could be an infinite amount of
            functions for which that frequency is also zero. However, we could
            also consider pointwise almost everywhere convergence, for which the
            proof is identical.} $(1-\widehat{\varphi_{\ep}})^n \to 0$ pointwise
            as $n \to \infty$. Thus, applying the Lebesgue's Dominated
            Convergence Theorem (LDCT)
        \begin{equation*}
            \lim_{n\to\infty}\int_{\R}|\widehat{f}(x)(1-(1-\widehat{\varphi_{\ep}}(x))^n
            - \widehat{f}(x)|^2dx = 
            \lim_{n\to\infty}\int_{\R}|\widehat{f}(x)|^2|1-\widehat{\varphi_{\ep}}(x)|^{2n}dx
            = 0.
        \end{equation*}
        Thus, $\lim_{n\to\infty}\Fourier{T_{\ep,n}^{-1}(T_{\ep}(f))} = 
        \widehat{f}$
        in $L^2(\R)$. Since the Fourier transform is an isometric 
        isomorphism in $L^2(\R)$, it holds that
        \begin{align*}
            f = \mathcal{F}^{-1}(\widehat{f}) &= 
            \mathcal{F}^{-1}\left({\lim_{n\to\infty}\Fourier{T_{\ep,n}^{-1}(T_{\ep}(f))}}\right)
            = 
            \lim_{n\to\infty}\mathcal{F}^{-1}\Fourier{T_{\ep,n}^{-1}(T_{\ep}(f))}
            = \lim_{n\to\infty}T_{\ep,n}^{-1}(T_{\ep}(f)),
        \end{align*}
        in $L^2(\R)$. This proves the statement.
        
        \item We know by the previous statement that
        \begin{equation*}
            \Fourier{T_{\ep,n}^{-1}(T_{\ep}(f))} = 
            \left(1-(1-\widehat{\varphi_{\ep}})^{n+1}\right)\widehat{f}, 
            \quad 
            \forall n \in \N.
        \end{equation*}
        Again note that by Young's inequality, $T_{\ep}(f) \in L^1(\R)$, for $f
        \in L^1(\R)$, and so $T_{\ep,n}^{-1}(T_{\ep}(f)) \in L^1(\R)$ for all $n
        \in \N$. We have that
        $\left(1-(1-\widehat{\varphi_{\ep}})^{n+1}\right)\widehat{f} \to
        \widehat{f}$ pointwise as $n \to \infty$. Moreover
        $|\left(1-(1-\widehat{\varphi_{\ep}})^{n+1}\right)\widehat{f}| \leq
        2|\widehat{f}|$, and being $\widehat{f} \in L^1(\R)$, we can apply the
        LDCT
        \begin{align*}
            \lim_{n\to\infty}T_{\ep,n}^{-1}(T_{\ep}(f))(x) &= 
            \lim_{n\to\infty}\mathcal{F}^{-1}\Fourier{T_{\ep,n}^{-1}(T_{\ep}(f))}(x)
            \\
            &= 
            \lim_{n\to\infty}\int_{\R}\Fourier{T_{\ep,n}^{-1}(T_{\ep}(f))}(\xi)e^{2\pi
                i 
                \xi x}d\xi \\
            &= 
            \lim_{n\to\infty}
            \int_{\R}\left(1-(1-\widehat{\varphi_{\ep}}(\xi))^{n+1}\right)\widehat{f}(\xi)e^{2\pi
            i \xi x}d\xi \\
            &= \int_{\R}\widehat{f}(\xi)e^{2\pi i\xi x}d\xi  \\
            &= \mathcal{F}^{-1}(\widehat{f})(x) = f(x). 
        \end{align*}
        Clearly, we have an almost everywhere results if the pointwise converge
        is also almost everywhere. 
        \item Let $u = \one$. It is known that $\widehat{\one} = \delta$.
        Therefore, $\widehat{T_{\ep,n}^{-1}(T_{\ep}(\one))} = 
        \delta\cdot (1-(1-\widehat{\varphi_{\ep}})^{n+1})$. Thus,
        given $\psi \in \Sw(\R)$,
        \begin{align*}
            \dotProduct{\widehat{T_{\ep,n}^{-1}(T_{\ep}(\one))}}{\psi}
            &= \dotProduct{\delta\cdot 
                (1-(1-\widehat{\varphi_{\ep}})^{n+1})}{\psi} 
            =
            \dotProduct{\delta}{(1-(1-\widehat{\varphi_{\ep}})^{n+1})\psi}
            = (1-(1-\widehat{\varphi_{\ep}})^{n+1})(0)\psi(0).
        \end{align*}
        Since $(1-(1-\widehat{\varphi_{\ep}})^{n+1}) \to 1$ pointwise as $n 
        \to \infty$,
        \begin{equation*}
            \lim_{n\to\infty}\dotProduct{\widehat{T_{\ep,n}^{-1}(T_{\ep}(\one))}}{\psi}
            = \dotProduct{\delta}{\psi},
        \end{equation*}
        and so $\widehat{T_{\ep,n}^{-1}(T_{\ep}(\one))} \to \delta =
        \widehat{\one}$ as $n \to \infty$ in the weak-* topology of $\Sw'$,
        i.e., in the sense of distributions. Since the Fourier transform is a
        homeomorphism in $\Sw'$ it follows that $T_{\ep,n}^{-1}(T_{\ep}(\one))
        \to \one$ as $n\to \infty$ in the sense of distributions. The case for
        $u = \delta$ follows by noting that
        \begin{equation*}
            \dotProduct{\widehat{T_{\ep,n}^{-1}(T_{\ep}(\delta))}}{\psi} = 
            \dotProduct{\one}{(1-(1-\widehat{\varphi_{\ep}})^{n+1})\psi}
            = \int_{\R}(1-(1-\widehat{\varphi_{\ep}})^{n+1})\psi, \quad \forall 
            \psi 
            \in \Sw(\R),
        \end{equation*}
        and then applying the LDCT. For the principal value distribution,
        $\widehat{\vp} = -i\pi \sign$ as a distribution, then, given $ \psi \in
        \Sw(\R)$,
        \begin{align*}
            \dotProduct{\widehat{T_{\ep,n}^{-1}(T_{\ep}(\vp))}}{\psi}
            &=\dotProduct{-i\pi 
                \sign}{(1-(1-\widehat{\varphi_{\ep}})^{n+1})\psi}
            = 
            -i\pi\int_{\R}\sign(\xi)(1-(1-\widehat{\varphi_{\ep}}(\xi))^{n+1})\psi(\xi)d\xi,
        \end{align*}
        and the result follows by applying again the LDCT. Since
        $\widehat{e^{2\pi i\id y}}= \delta_{y}$, i.e., the Dirac delta centered
        at $y \in \R$ the results follows from the previous proof. The
        convergence of the linear combination and translations of these
        distributions follow immediately.
        
        \item First we prove that for the claims to hold, it is sufficient 
        to prove that for any $\psi \in \Sw(\R)$ then $\psi (1-
        \widehat{\varphi_{\ep}})^n \to 0$ as $n \to \infty$ in the Schwartz
        topology. Indeed, let $f \in \Sw(\R)$ and suppose it holds that
        $\widehat{f}(1-\widehat{\varphi_{\ep}})^n \to 0$ in the Schwartz
        topology as $n \to \infty$. Since $\widehat{f} \in \Sw(\R)$ and applying
        the same steps as in the previous proofs,
        \begin{equation*}
            \widehat{T_{\ep,n}^{-1}(T_{\ep}(f))} = 
            \hat{f}\left(1-(1-\widehat{\varphi_{\ep}})^{n+1}\right) \to 
            \hat{f}, \quad n \to \infty
        \end{equation*}
        in the Schwartz topology by assumption. Since the Fourier transform is a
        homeomorphism on $\Sw(\R)$, then 
        \begin{align*}
            f = \mathcal{F}^{-1}(\widehat{f}) &=  
            \mathcal{F}^{-1}\left({\lim_{n\to\infty}\Fourier{T_{\ep,n}^{-1}(T_{\ep}(f))}}\right)
                                           =  
            \lim_{n\to\infty}\mathcal{F}^{-1}\Fourier{T_{\ep,n}^{-1}(T_{\ep}(f))}
            = \lim_{n\to\infty}T_{\ep,n}^{-1}(T_{\ep}(f)),
        \end{align*}
        in the Schwartz topology. If $f \in \Sw'(\R)$, being $\varphi$ even,
        \begin{equation*}
            \dotProduct{\widehat{T_{\ep,n}^{-1}(T_{\ep}(f))}}{\psi} = 
            \dotProduct{T_{\ep,n}^{-1}(T_{\ep}(f))}{\widehat{\psi}}
            = \dotProduct{f}{T_{\ep,n}^{-1}(T_{\ep}(\widehat{\psi}))}, 
            \quad \forall \psi 
            \in \Sw(\R).
        \end{equation*}
        Since $T_{\ep,n}^{-1}(T_{\ep}(\widehat{\psi})) \to \widehat{\psi}$ as $n
        \to \infty$ in the Schwartz topology, and tempered distributions are
        continuous functionals on $\Sw(\R)$, then
        $\dotProduct{\widehat{T_{\ep,n}^{-1}(T_{\ep}(f))}}{\psi} \to
        \dotProduct{f}{\widehat{\psi}} = \dotProduct{\widehat{f}}{\psi}$ as $n
        \to \infty$ in the weak-* topology, i.e., in the tempered distribution
        sense. Since the Fourier transform is a homeomorphism on $\Sw'(\R)$ the
        result follows as previously noted.
        
        So we proceed now to prove that for any $\psi \in \Sw(\R)$ then $\psi
        (1- \widehat{\varphi_{\ep}})^n \to 0$ as $n \to \infty$ in the Schwartz
        topology if $\varphi$ is the Gaussian Kernel. Assume without loss of
        generality that $\ep = 1$. Then $\widehat{\varphi}(\xi) = e^{-\xi^2}$
        for all $ \xi \in \R$. Let $h(x)=1-e^{-x^2}$ for all $x \in \R$. We need
        to prove that for all $\alpha,\beta \in \N_0$,
        \begin{equation}\label{eq:ConditionSchwartz}
            \lim_{n\to\infty}\sup_{x\in\R}\left|x^{\alpha}D^{\beta}(\psi 
            h^n)(x)\right| = 0,
        \end{equation}
        where $D$ is the differential operator, i.e., $D^{\beta}f = 
        f^{(\beta)}$.
        
        Fix $\alpha,\beta \in \N_0$. For simplicity of notation, we will denote
        both the function $f$ and its value at $x$ as $f(x)$. First, it is not
        difficult to check by induction that for any $m \in \N$, it holds
        $D^{m}h(x) = P(m,x)e^{-x^2}$ where  $P(m,x)$ is a polynomial of degree
        up to $m$. Applying Faa di Bruno's formula \cite{roman1980formula} for
        $m \leq n$, it holds that
        \begin{equation}\label{eq:FaaDiBrunoExponential}
            D^m(h(x))^n = \sum \frac{m!}{k_1!1!^{k_1}k_2!2^{k_2}\dots 
                k_m!m!^{k_m}}\frac{n!}{(n-\bar{k})!}(h(x))^{n-\bar{k}}\prod_{j=1}^{m}(D^jh(x))^{k_j},
        \end{equation}
        where the sum is taken over all tuples $(k_1,\dots,k_m)$ such that
        $\sum_{j=1}^mjk_j = m$ and $\bar{k} := \sum_{j=1}^mk_j$. We define for
        any $m \in \N$, $\pi(m) := \frac{m!}{k_1!1!^{k_1}k_2!2^{k_2}\dots
            k_m!m!^{k_m}}$. Note that $D^jh(x) = P(j,x)e^{-x^2}$ and therefore,
            by the properties of the exponential,
        $\prod_{j=1}^{m}(D^jh(x))^{k_j} = P(m!,x)e^{-\bar{k}x^2}$. Thus, we can
        write Equation \eqref{eq:FaaDiBrunoExponential} as
        \begin{equation}\label{eq:DerivativeOfNthPower}
            D^m(h(x))^n = \sum 
            \pi(m)\frac{n!}{(n-\bar{k})!}\left(1-e^{-x^2}\right)^{n-\bar{k}}P(m!,k)e^{-\bar{k}x^2},
        \end{equation}
        and note that the only dependence on $n$ occurs in 
        $\frac{n!}{(n-\bar{k})!}(1-e^{-x^2})^{n-\bar{k}}$. Also note that
        $\frac{n!}{(n-\bar{k})!} \leq n^{\bar{k}}$ for all $n,\bar{k} \in 
        \N$ such that $n \geq \bar{k}$, which we are going to assume from 
        now on. Applying the Leibniz formula of differentiation
        to Equation \eqref{eq:ConditionSchwartz} and inserting Equation
        \eqref{eq:DerivativeOfNthPower}, we get
        \begin{align*}
            &\sup_{x\in \R}\left|x^{\alpha}D^{\beta}(\psi 
            h^n)(x)\right| 
             \leq \sum_{m=0}^{\beta}\sum 
            \binom{\beta}{m}\pi(m) \sup_{x\in \R}
            \left|x^{\alpha}D^{\beta-m}\psi(x)
            \frac{n!}{(n-\bar{k})!}(1-e^{-x^2})^{n-\bar{k}}P(m!,x)e^{-\bar{k}x^2}\right|,
        \end{align*}
        since the sums are finite and do not depend on $n$, we just need to 
        show that for any $m \in \{0,\dots,\beta\}$ it holds 
        \begin{align*}
            &\lim_{n\to\infty}\sup_{x\in \R}E_{n,m}(x) = 0,
        \end{align*}
        where
        \begin{align*}
            &E_{n,m}(x):=
            \left|x^{\alpha}D^{\beta-m}\psi(x)
            \frac{n!}{(n-\bar{k})!}(1-e^{-x^2})^{n-\bar{k}}P(m!,x)e^{-\bar{k}x^2}\right|.
        \end{align*}
        It is clear that $E_{n,m}(0) = 0$. We proceed by breaking $\R$ in three
        intervals. The first $|x|\leq \delta$ for any $\delta < 1$, the second
        $\delta < |x| \leq R$ for some $R \in \R$ and the third $|x| > R$. We
        will show that neither $\delta$ nor $R$ depend on $n$. We know that, by
        Taylor's formula, for a sufficiently small $\delta$, it holds that
        $1-e^{-x^2} \leq x^2$ for all $|x| \leq \delta$. Being $x^2$ strictly
        increasing, $x^2 \leq \delta^2$. Moreover, $\psi \in \Sw(\R)$ which
        implies 
        \begin{equation*}
            C_{\alpha,\beta,m} := \sup_{x\in \R}
            \left|x^{\alpha}D^{\beta-m}\psi(x)P(m!,x)e^{-\bar{k}x^2}\right| 
            < \infty.
        \end{equation*}
        Thus
        \begin{equation*}
            \sup_{|x|\leq \delta}E_{n,m}(x) \leq 
            C_{\alpha,\beta,m}\frac{n!}{(n-\bar{k})!}\delta^{2(n-\bar{k})} 
            \leq C_{\alpha,\beta,m}n^{\bar{k}}\delta^{2(n-\bar{k})}.
        \end{equation*}
        But the exponential decrease dominates polynomial increase, i.e., 
        $\lim_{n\to\infty}n^{\bar{k}}\delta^{2n}=0$, and 
        so
        \begin{equation*}
            \lim_{n\to\infty}\sup_{|x|\leq \delta}E_{n,m}(x) = 0.
        \end{equation*}
        Fix now $R \in \R$. Then, for $\delta \leq |x| \leq R$ it holds that
        $0<h(x) = 1-e^{-x^2} < 1$. Take $\gamma = \max_{\delta \leq |x| \leq
        R}h(x) < 1$, which clearly exists because $h$ is even and monotonically
        increasing for $x \geq 0$. Then we arrive at the same conclusions as
        above by the same argument. Thus $\lim_{n\to\infty}\sup_{\delta<|x|\leq
        R}E_{n,m}(x) = 0$. So we just need to prove that $R$ does not depend on
        $n$, because if it does not, being $\psi \in \Sw(\R)$ and $1-h^n \in
        \Sw(\R)$ then $\psi(1-h^n) \in \Sw(\R)$ and so $\psi h^n \in \Sw(\R)$,
        which implies that $\lim_{|x|\to\infty}E_{n,m}(x) = 0$, but the limit
        depends on $n$, so if we make it independent of $n$ we are done by
        splitting the domain as previously noted. For this purpose, note that
        being $\psi \in \Sw(\R)$ it holds that for any $m \in \N_0$,
        \begin{equation*}
            \lim_{|x|\to\infty}\left|x^{\alpha}D^{\beta-m}\psi(x)P(m!,x)\right|
            = 0.
        \end{equation*}
        Consider the function 
        $g_x(n):=\frac{n!}{(n-\bar{k})!}(1-e^{-x^2})^{n-\bar{k}}e^{-\bar{k}x^2}
        \geq 0$ for all $n \in \N$ with $n \geq \bar{k}$ and all $x \in \R$.
        Observe that all terms of $g$ are strictly positive, that 
        $\frac{n!}{(n-\bar{k})!} 
        \leq n^{\bar{k}}$, and also $(1-e^{x^2})^{n-\bar{k}} \leq 
        e^{-(n-\bar{k})e^{-x^2}}$ because $1-x \leq e^{-x}$
        for all $x \in [0,1]$, and $e^{-x^2} \in [0,1]$. Therefore, $g_x(n) 
        \leq v_x(n):=
        n^{\bar{k}}e^{-(n-\bar{k})e^{-x^2}}e^{-\bar{k}x^2}$.
        But by calculus it is not difficult to check that $v_x$ has a global
        maximum at $n = \bar{k}e^{x^2}$, which implies that
        \begin{equation*}
            g_x(n) \leq \bar{k}^{\bar{k}}e^{-\bar{k}}e^{\bar{k}e^{-x^2}}, 
            \quad \forall 
            n \geq \bar{k}.
        \end{equation*}
        Therefore, the limit as $|x| \to \infty$ of $g_x(n)$ exists and is 
        less than $\bar{k}^{\bar{k}}e^{-\bar{k}}$ for all $n \in \N$ with 
        $n \geq \bar{k}$. Thus,
        \begin{equation*}
            0\leq \lim_{|x|\to\infty}E_{n,m}(x) \leq 
            \lim_{|x|\to\infty}\left|x^{\alpha}D^{\beta-m}\psi(x)P(m!,x)\right|
            \bar{k}^{\bar{k}}e^{-\bar{k}}e^{-\bar{k}e^{-x^2}} = 0\cdot 
            \bar{k}^{\bar{k}}e^{-\bar{k}} = 
            0.
        \end{equation*}
        That is, given $\eta > 0$ we can choose $R = R(\eta) > 0$, $\delta 
        > 0$ and $N = N(\eta,\alpha,\beta,m) \in \N$ such that for all $n 
        \geq N$,
        \begin{equation*}
            \sup_{x\in \R}E_{n,m}(x) \leq \sup_{|x|\leq \delta}E_{n,m}(x) + 
            \sup_{\delta < |x| \leq R}E_{n,m}(x) + \sup_{|x| > R}E_{n,m}(x) 
            < \eta.
        \end{equation*}
        This finishes the proof of the Theorem.
    \end{enumerate}
\end{proof}

We want to remark that the last statement of Theorem
\ref{thm:GeneralizationToSchwartzFunctionsAndDistributions} is a new formula for
the inverse of the Weierstrass Transform, i.e., the operator
$\lim_{n\to\infty}T_{\ep,n}^{-1}$ is the inverse for the Weierstrass Transform
with parameter $\ep > 0$ for any $f \in \Sw(\R)$ or $f \in \Sw'(\R)$. Note that
$L^p(\R)$ is injected in $\Sw'(\R)$ for any $p \in [1,\infty]$, as well
the set $\{v \in L^1_{loc}(\R) \mid |v(x)| \leq C(1+|x|^2)^{N/2} \text{ for some
$N \in \N$ and $C > 0$}\}$. Also note that the Delta, Heaviside and Principal
Value distributions are in $\Sw'(\R)$ and so any distribution with compact
support. Therefore, Theorem
\ref{thm:GeneralizationToSchwartzFunctionsAndDistributions} presents a result
that can be used in several situations which is numerically simple to compute.
To end this section, note that if $f \in \Sw(\R)$ of $f \in \Sw'(\R)$ and
$\varphi$ is the Gaussian Kernel, then, in their respective topologies
\begin{align*}
    f = \lim_{n\to\infty}T_{\ep,n}^{-1}(T_{\ep}(f)) &= 
    \lim_{n\to\infty}\sum_{k=0}^{n}(-1)^{k}\binom{n+1}{k+1}T_{\ep}^{k+1}(f) 
    \\
    &= 
    \lim_{n\to\infty}\sum_{k=1}^{n+1}(-1)^{k+1}\binom{n+1}{k}T_{\ep}^{k}(f) 
    \\
    &= 
    \lim_{n\to\infty}\sum_{k=0}^{n+1}(-1)^{k+1}\binom{n+1}{k}T_{\ep}^{k}(f) 
    -f,
\end{align*}
and so $\lim_{n\to\infty}\sum_{k=0}^{n+1}(-1)^{k+1}\binom{n+1}{k}T_{\ep}^{k}(f)
= \lim_{n\to\infty}-(\id - T_{\ep})^{n+1}(f) = 0$ in their respective topologies
for any $f \in \Sw(\R)$ or $f \in \Sw'(\R)$, and so $\lim_{n\to\infty}(\id -
T_{\ep})^{n} \equiv 0$, i.e., the zero operator on $\Sw(\R)$ or $\Sw'(\R)$. Note
the notation $(\id-T_{\ep})^{n}$ is the $n$-th folding composition of $\id -
T_{\ep}$.

In Section \ref{sec:Applications} we show the numerical validation of Theorem
\ref{thm:GeneralizationToSchwartzFunctionsAndDistributions}.

\subsection{Discussion about polynomials in the Euclidean 
    space}\label{subsec:PolynomialsRd}
This small section discusses how most of the results polynomials in the real
line extend to polynomials in the Euclidean space. Note that the Binomial
Theorem, as well as the Pascal Triangle and the proofs by induction extend
nicely if we use multi-index notation. Indeed, let $d \in \N$ and $\alpha,\beta
\in \N_0^{d}$, and $x,y \in \R^d$. If we define
\begin{align*}
    &x^{\alpha} := x_1^{\alpha_1}\dots x_d^{\alpha_d}, \quad &\alpha \pm 
    \beta 
    := 
    (\alpha_1 \pm \beta_1, \dots, \alpha_d \pm \beta_d) \\
    &\alpha \leq \beta \iff \alpha_i \leq \beta_i, \forall i \in 
    \{1,\dots,d\}, 
    \quad &\alpha! := \alpha_1! \cdot \alpha_2! \dots \alpha_d! \\
    &\binom{\alpha}{\beta} := \frac{\alpha!}{\beta!(\alpha-\beta)!} \quad & 
    |\alpha| := |\alpha_1 + \dots + \alpha_d|,
\end{align*}
then it holds that 
\begin{equation*}
    (x\pm y)^{\alpha} = \sum_{\nu \leq 
        \alpha}(\pm 
    1)^{|\alpha-\nu|}\binom{\alpha}{\nu}x^{\nu}y^{\alpha-\nu}.
\end{equation*}
We define $\Poly_{n}(\R^d)$ as the vector space of $n \in \N_0$ order
polynomials in $\R^d$, that is,
\begin{equation*}
    p \in \Poly_{n}(\R^d) \iff p(x)=\sum_{|\alpha|\leq 
        n}a_{\alpha}x^{\alpha}, 
    \quad \forall x \in \R^d.
\end{equation*}
Note that a similar result to Definition \ref{def:PowersOfPolynomialIntegration}
holds in $\R^d$. Indeed, let $\varphi \in \SM(\R^d)$ be even. Then 
\begin{equation*}
    \int_{\R^d}x^{\alpha}\varphi_{\ep}(x)dx 
    = \ep^{d|\alpha|}\int_{\R^d}x^{\alpha}\varphi(x)dx = \begin{cases}
        \ep^{d|\alpha|}c_{\alpha} & |\alpha| \text{ is even or zero} \\
        0 & |\alpha| \text{ is odd}
    \end{cases},
\end{equation*}
where $c_{\alpha} := \int_{\R^d}x^{n}\varphi(x)dx$, and clearly $c_{0} = 1$.

Fix then any even $\varphi \in \SM(\R^d)$ and $\ep > 0$. In either case, observe
it is in fact immediate to generalize certain previous results. Indeed, the
proof of Theorem \ref{thm:MollificationOfPolynomialsRealLine} is extrapolated as
is to $\R^d$, from which we get that 
\begin{equation*}
    (p * \varphi_{\ep})(x) = \sum_{\alpha \leq 
        n}\sum_{\substack{\beta \leq \alpha \\ |\alpha-\beta| \in 
            2\N_0}}\binom{\alpha}{\beta}a_{\alpha}x^{\beta}\ep^{d|\alpha-\beta|}c_{\alpha-\beta},
    \quad (x \in \R^d).
\end{equation*}
That is, a similar result to Equation \eqref{eq:PolyMollificationRealLine}.
Clearly, then
\begin{align*}
    (p * \varphi_{\ep})(x) &= \sum_{\alpha \leq n}a_{\alpha}x^{\alpha} +
    \sum_{\alpha \leq n}\sum_{\substack{\beta < \alpha \\ |\alpha-\beta| \in
    2\N}}\binom{\alpha}{\beta}a_{\alpha}x^{\beta}\ep^{d|\alpha-\beta|}c_{\alpha-\beta}
   =: p(x) + p_1(x), \quad (x \in \R^d).
\end{align*}
which is similar to the results of equation
\eqref{eq:RelationShipPolyMolliPoly} in Corollary
\ref{cor:CorollaryPolynomials}. That is, the convolution of an $n$ degree
polynomial in $\R^d$ with even Schwartz convolution  kernel, results in the sum
of the polynomial itself with a polynomial of degree $n-2$. This therefore
implies that the results of Corollary \ref{cor:PolynomialsMapping} can be
extrapolated, by noting that $|\alpha| = 1 \iff \alpha \in \{ (1,\dots ,0),
(0,1\dots,0), \dots, (0,\dots,1)\}$, and therefore $x \mapsto x^{\alpha} = x_i$
for some $i \in \{1,\dots,d\}$. Thus,
$\int_{\R^d}(x-y)^{\alpha}\varphi_{\ep}(y)dy = x_i -
\int_{\R^d}x_i\varphi_{\ep}(y)dy = s_i$ for some $i \in \{1,\dots,d\}$ and all
$x \in \R^d$. Thus if $p \in \Poly_{1}(\R^d)$ we get
\begin{equation*}
    (p * \varphi_{\ep})(x) = \sum_{|\alpha|\leq 
        1}c_{\alpha}\int_{\R^d}(x-y)^{\alpha}\varphi_{\ep}(y)dy = 
    \sum_{|\alpha|\leq 1}c_{\alpha}x^{\alpha} = p(x), \quad (x \in \R^d),
\end{equation*}
that is, polynomials in $\Poly_{1}(\R^d)$ and $\Poly_{0}(\R^d)$ are invariant to
the convolution with even Schwartz convolution  kernel.
Thus, we can reach the same conclusions as in Theorem 
\ref{thm:ConvolutionIsAnIsomorphismInPolynomials} with an identical proof. 
That 
is, convolution with an
even convolution kernel is an isomorphism of $\Poly_{n}(\R^d)$ onto itself 
for all
$n \in \N_0$, and we can compute the inverse of the convolution similar to 
equations
\eqref{eq:InverseOfConvolution} and \eqref{eq:InverseOfConvolutionOdd}.
Indeed, if $n$ is even the inverse of $p \in \Poly_{n}(\R^d) \mapsto 
T_{\ep}(p) 
= p * \varphi_{\ep} \in C^{\infty}(\R^d)$ with $\varphi \in \SM(\R^d)$ even 
is
\begin{equation*}
    T_{\ep}^{-1}(p) = \sum_{|\alpha|\leq n/2}\sum_{\beta \leq
    \alpha}(-1)^{|\alpha+\beta|}\binom{\alpha}{\beta}T_{\ep}^{|\alpha-\beta|}(p),
    \quad (p \in \Poly_{n}(\R^d)),
\end{equation*}
while if $n$ is odd
\begin{equation*}
    T_{\ep}^{-1}(p) = \sum_{|\alpha|\leq \frac{n-1}{2}}\sum_{\beta \leq 
        \alpha}(-1)^{|\alpha+\beta|}\binom{\alpha}{\beta}T_{\ep}^{|\alpha-\beta|}(p),
    \quad (p \in \Poly_{n}(\R^d)).
\end{equation*}

\section{Applications}\label{sec:Applications}
In this section we present two particular applications of the results
in Section \ref{sec:TheoreticalResults}. First we validate numerically
Theorem \ref{thm:ConvolutionIsAnIsomorphismInPolynomials}, and then
we apply the results of Theorem 
\ref{thm:GeneralizationToSchwartzFunctionsAndDistributions} to an 
application in the reconstruction of a sinusoidal function.
\begin{figure}[pos = t]
    \centering 
    \includegraphics[width=\linewidth]{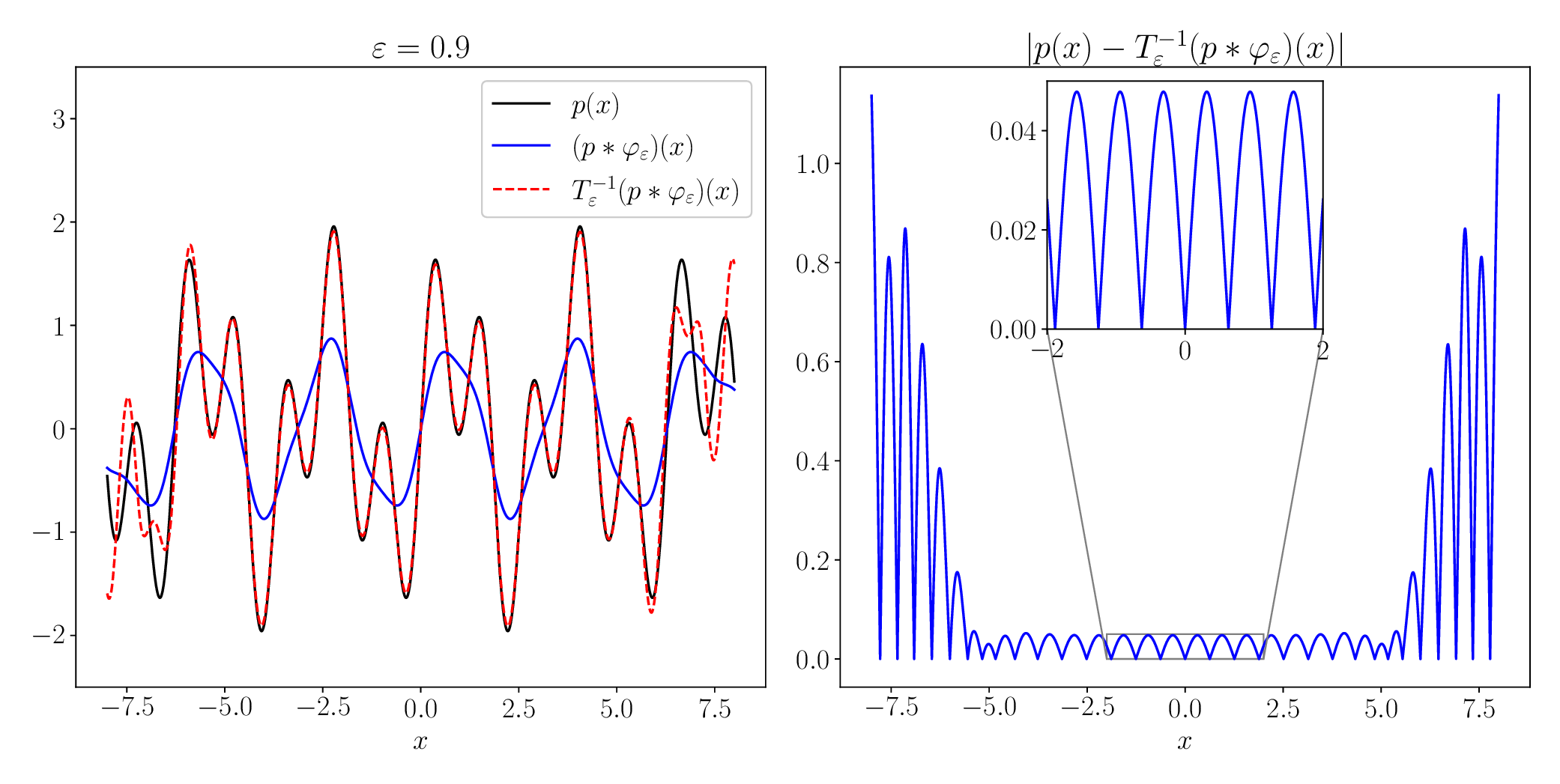}
    \caption{Representation of a numerical validation of Theorem 
        \ref{thm:ConvolutionIsAnIsomorphismInPolynomials}. The left picture
        represents as a solid black line the original polynomial $p$ which is
        the $n=50$ order Taylor polynomial of the function $x \in \R \mapsto
        \sin(5x) + \sin(3x)$. It also represents as a blue solid line its
        convolution $p * \varphi_{\ep}$ with convolution kernel $s\mapsto
        \varphi(x) = \ind_{(-1,1)}\exp(-1/(1-x^2))$ and parameter $\ep = 0.9$.
        The dashed red line represents $T_{\ep}^{-1}(p * \varphi_{\ep})$
        computed using equation \eqref{eq:InverseOfConvolution}. The right
        picture represents the error between the original polynomial and the
        inverse of the convolution.}
    \label{fig:InversionNumericalExperiment}
\end{figure}

\subsection{Validation of Theorem \ref{thm:ConvolutionIsAnIsomorphismInPolynomials}}
In order to numerically validate Theorem
\ref{thm:ConvolutionIsAnIsomorphismInPolynomials}, we present Figure
\ref{fig:InversionNumericalExperiment}. The polynomial is an $n=50$ Taylor
polynomial of the function $x \in \R \mapsto \sin(3x) + \sin(5x)$. Due to the
fact that, numerically speaking, the domain of the original polynomial is always
bounded, it can be seen that at the ends of the domain the inverse of the
convolution and the original polynomial differs the most.  This is a fact on how
convolution is computed numerically. Due to being the functions actually finite
sequences, a zero padding is needed in order to be able to compute the 
convolution numerically. Then, once computed, we take back the same number of 
samples as the original sequence in order to compare it with the original
function in its domain of definition. Nevertheless, it can be seen that, once
the effect of the edges is negligible, that Figure
\ref{fig:InversionNumericalExperiment} validates numerically Theorem
\ref{thm:ConvolutionIsAnIsomorphismInPolynomials}, recovering $p$ from $p *
\varphi_{\ep}$ by applying Equation \eqref{eq:InverseOfConvolution}. Note that
the used convolution kernel has compact support and zeros in its Fourier
transform (cf. Theorem
\ref{thm:GeneralizationToSchwartzFunctionsAndDistributions}).

\subsection{Applications to Signal Processing}\label{subsec:ApplicationSignalProcessing}

We are going to consider the problem of deconvolution of an arbitrary function.
Suppose that the convolution kernel is the Gaussian Kernel, $\varphi \in
\Sw(\R)$. Suppose then, that some function $f : \R \to \R$ has been convolved
with $\varphi_{\ep}$ for some $\ep > 0$. Suppose further that $f$ satisfies any
condition of Theorem
\ref{thm:GeneralizationToSchwartzFunctionsAndDistributions}, Then, by letting $n
\to \infty$ in Equation \eqref{eq:InverseOfConvolution} we could recover $f$ (in
the appropriate convergence).

The traditional approach to deconvolution consists on working on the 
Fourier transformed space. For example, under the appropriate conditions,
for $\varphi$ the Gaussian Kernel and $f \in \Sw(\R)$ it holds that
$\widehat{f * \varphi}(\xi) = \widehat{f}\widehat{\varphi}(\xi) \Longrightarrow
f = \mathcal{F}^{-1}\left(\widehat{f*\varphi}e^{\id^2}\right)$. Note
this is ill-posed numerically speaking, since $x \mapsto e^{x^2}$ growths
exponentially fast, and in particular, this function does not posses
inverse Fourier transform. However, by using the results of Theorem
\ref{thm:GeneralizationToSchwartzFunctionsAndDistributions}, we can reconstruct
$f$ directly from the samples of $f * \varphi$, without the need of passing to
the frequency domain.

Consider then, the function $t \in \R \mapsto f(t) = \sin(5t) + \sin(3t)$ 
(we denote $t$ for time)
as in Figure \ref{fig:InversionNumericalExperiment} and
let $\varphi$ be the Gaussian convolution kernel. Since $f$ defines
a regular tempered distribution, the assumptions of Theorem 
\ref{thm:GeneralizationToSchwartzFunctionsAndDistributions} hold.
\begin{figure}[pos = t]
    \centering
    \includegraphics[width=\linewidth]{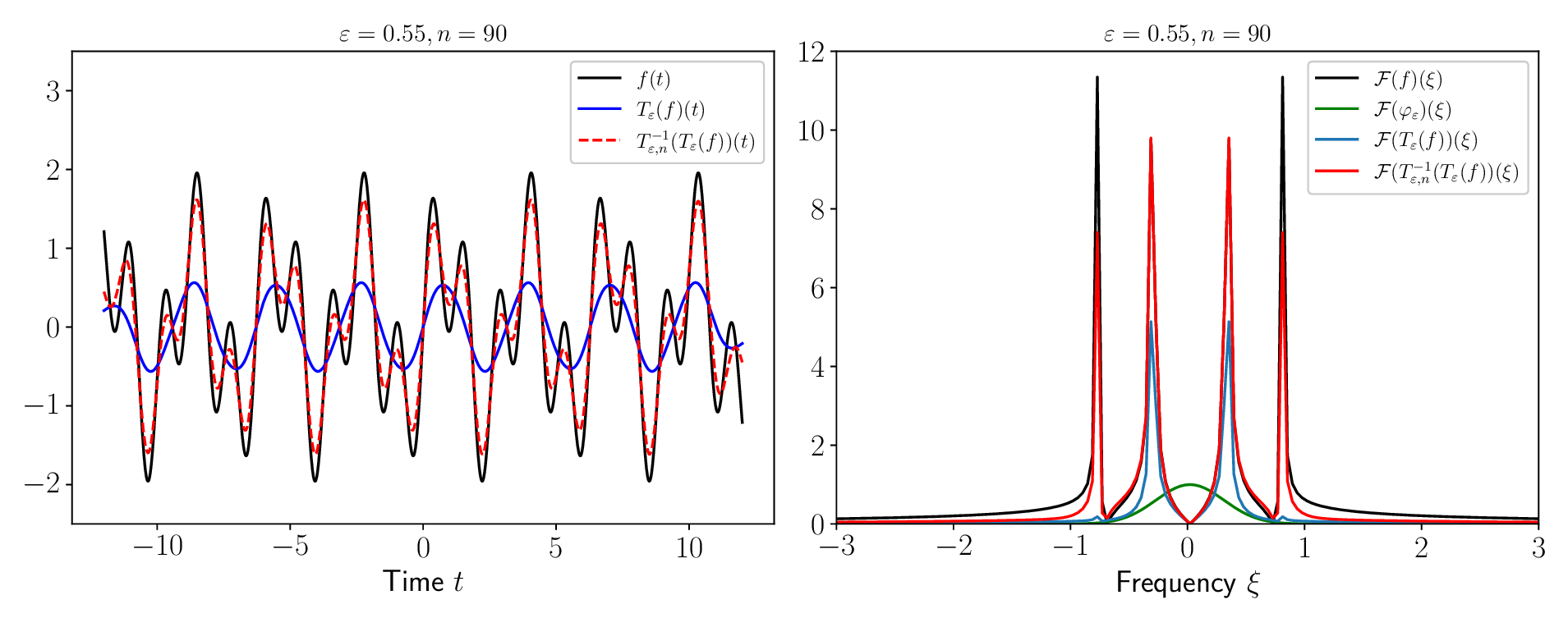}
    \caption{Representation of the numerical validation of Theorem 
        \ref{thm:GeneralizationToSchwartzFunctionsAndDistributions}
        under ideal conditions. Left 
        picture represents in black the original function $t \mapsto f(t) = 
        \sin(5t)+\sin(3t)$, in blue its convolution $T_{\ep}(f)$ with the Gaussian 
        Kernel $\varphi$ with $\ep = 0.55$, and as dashed red line its
        reconstruction $T_{\ep,n}^{-1}(T_{\ep}(f))$ using Equation
        \eqref{eq:InverseOfConvolutionArbitrary} for $n= 90$. Right picture
        represents in black the Fast Fourier Transform (FFT) of $f$, in green
        the FFT of $\varphi$, in blue the FFT of $T_{\ep}(f)$ and in red the FFT of
        the reconstructed function $T_{\ep,n}^{-1}(T_{\ep}(f))$.}
    \label{fig:ApplicationSignalProcessingOne}
\end{figure}
Note $T_{\ep}(f) = f * \varphi_{\ep}$ for $\ep > 0$ is 
the available data. In Figure \ref{fig:ApplicationSignalProcessingOne} we can 
see the results of applying Theorem
\ref{thm:GeneralizationToSchwartzFunctionsAndDistributions} to $T_{\ep}(f)$.
As it can be seen, the signal reconstruction is almost perfect. Theoretically
speaking, we would need to let $n \to \infty$ for (distributional) convergence,
but at $n \geq 100$ the numerical errors start to affect the method, and the 
recovery of the signal for $n=90$ is sufficiently good, numerically speaking,
as can be seen from the time and frequency domains. Indeed, by looking 
at the spectra of $f$, $T_{\ep}(f)$ and $T_{\ep,n}^{-1}(T_{\ep}(f))$, it is 
well validated that the method works, since we were capable of recovering 
(numerically speaking) the original function.

Now we consider the situation depicted in Figure
\ref{fig:ApplicationSignalProcessingTwo}. We suppose that $T_{\ep}(f)$
has been modified with additive noise $\eta : \R \to \R$ with zero mean, and
variance $0.5$. Note that the signal to noise ratio is less than one,
as it can be seen from the top and bottom left pictures. However,
by applying Equation \eqref{eq:InverseOfConvolutionArbitrary} to
$T_{\ep}(f)+\eta$, it can be seen from the top right picture that we still
recover (with amplified noise) the frequency contents that were once removed by
the convolution kernel. This can be ``heuristically'' be justified, by noting 
that if $T_{\ep,n}^{-1}$ is the inverse (in the limit) of the convolution 
of ``sufficiently nice'' functions, then this operator must amplify noise,
being a deconvolution. Thus, after computing Equation
\eqref{eq:InverseOfConvolutionArbitrary} it is clear that the original function
was (at least) composed of higher frequencies than the spectrum of
$T_{\ep}(f)+\eta$ shows. Thus, we can apply a low pass filter $h:\R\to\R$ to 
the recovered noisy signal in order to remove the undesired noise. This is
shown in the bottom pictures. As it can be seen, we have effectively recovered
a close version to the original signal, both in the time and frequency domain,
even if noise was present, with the particularity that the signal to noise
ratio was lower than one.

\begin{figure}[pos = t]
    \centering
    \includegraphics[width=\linewidth]{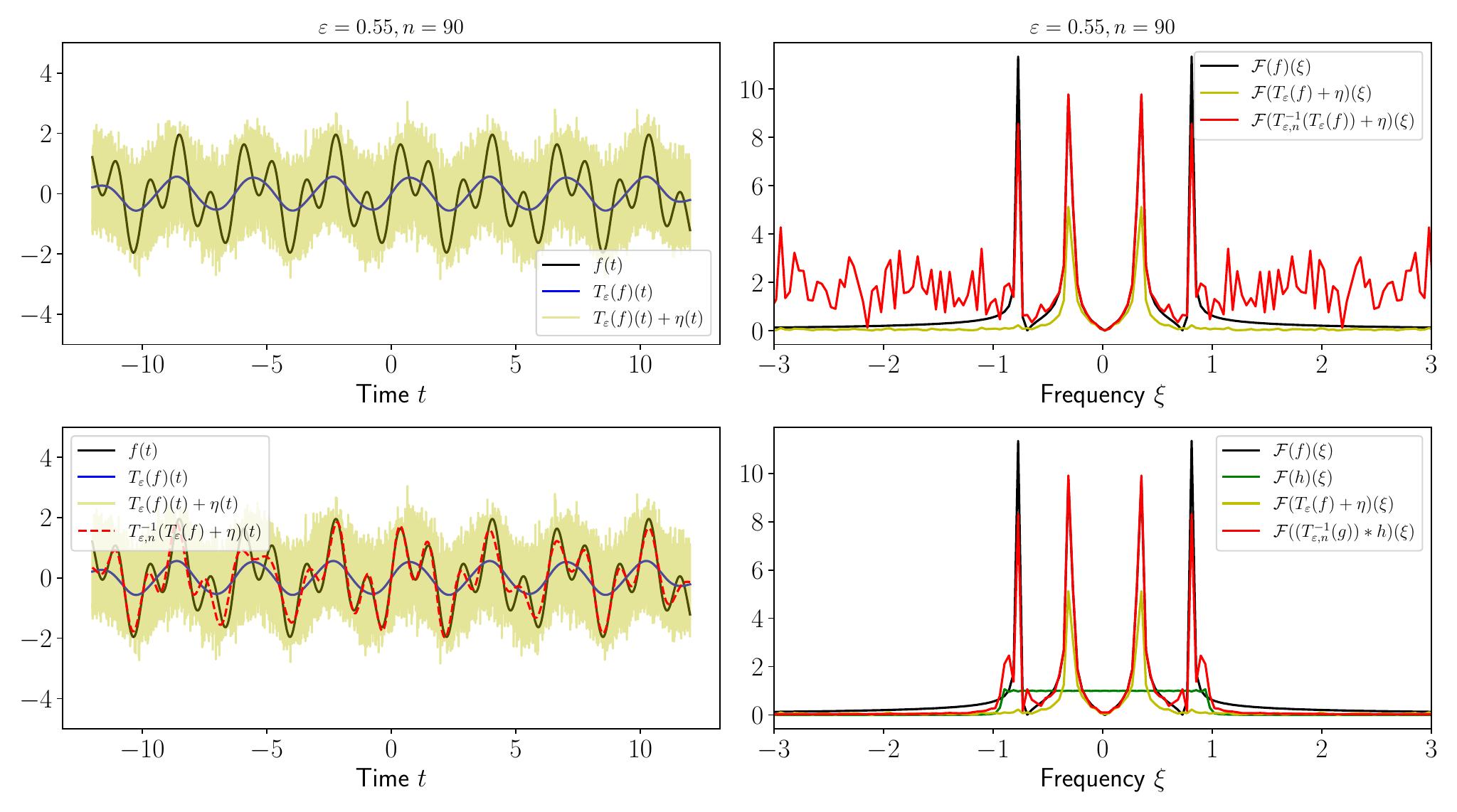}
    \caption{Representation of the numerical validation of Theorem
        \ref{thm:GeneralizationToSchwartzFunctionsAndDistributions} under 
        non ideal conditions. Top and bottom left pictures
        represent in black the original function 
        $t \mapsto f(t) = \sin(5t)+\sin(3t)$, with a blue line $T_{\ep}(f)$ with
        $\ep = 0.55$ and $\varphi$ the Gaussian Kernel,
        as a yellow line $T_{\ep}(f) + \eta$, where $\eta$ is noise generated
        using a normal distribution with zero mean and variance 0,
        and with a red line $T_{\ep,n}^{-1}(T_{\ep}(f)+\eta)$ where $n =90$.
        Top right picture depicts as a black line the FFT of $f$,
        as a yellow line the FFT of $T_{\ep}(f)+\eta$ and as a red line 
        $T_{\ep,n}^{-1}(T_{\ep}(f)+\eta)$, i.e., the application of 
        Equation \eqref{eq:InverseOfConvolutionArbitrary} to $T_{\ep}(f)+\eta$.
        The bottom right picture shows with a black line the FFT of $f$,
        with a green line the FFT of the used filter $t \mapsto h(t) =
        2\sinc(2t)$,  with a yellow line the FFT of $T_{\ep}(f)+\eta$ and with
        a red line the FFT of $T_{\ep,n}^{-1}(T_{\ep}(f)+\eta) * h$, that is, $g
        := T_{\ep}(f) + \eta$.}
    \label{fig:ApplicationSignalProcessingTwo}
\end{figure}

\section{Conclusions}\label{sec:Conclusions}
This work established a rigorous and computationally robust framework for the
exact inversion of convolution operators with even Schwartz kernels. We first
proved that these operators act as automorphisms on the vector space of
finite-order polynomials, yielding an explicit, finite-sum inversion formula
based solely on recursive forward convolutions.
By extending this algebraic foundation to infinite-dimensional spaces, we
derived a novel iterative formula for recovering functions and tempered
distributions in their respective topologies. This approach provides a new
alternative to traditional differential inversion of the Weierstrass transform,
bypassing the numerical ill-posedness inherent in such operators. 
Numerical validation confirmed the feasibility of the method, demonstrating
successful signal recovery for sinusoidal functions in both noiseless and low
signal-to-noise ratio --- less than unity --- environments. These results
confirm that the recursive convolution framework is not only theoretically exact
but also highly resilient for practical signal processing and computational
applications

\section{Acknowledgments}
This work has been supported by the FPU program of the Ministry of science,
innovation and universities of Spain. The author thanks Juan F. Jim\'enez for
his helpful feedback that contributed to the refinement of this paper.

\printcredits

\bibliographystyle{cas-model2-names}

\bibliography{cas-refs}



\end{document}